\documentclass{article}

\usepackage[margin=1in]{geometry}

\usepackage{soul}
\usepackage{amsmath,amsthm,mathtools,amssymb}
\usepackage{algorithm}
\usepackage{algorithmic}
\usepackage[usenames,svgnames,xcdraw,table]{xcolor}
\definecolor{DarkBlue}{rgb}{0.1,0.1,0.5}
\definecolor{DarkGreen}{rgb}{0.1,0.5,0.1}
\usepackage{hyperref}
\hypersetup{
   colorlinks   = true,
 linkcolor    = DarkBlue, 
 urlcolor     = DarkBlue, 
	 citecolor    = DarkGreen 
}

\newtheorem{theorem}{Theorem}
\newtheorem{lemma}{Lemma}
\newtheorem{claim}[lemma]{Claim}
\newtheorem{remark}{Remark}
\newtheorem{proposition}{Proposition}
\newtheorem*{proposition*}{Proposition}

\newtheorem{definition}{Definition}

\let\originalleft\left
\let\originalright\right
\renewcommand{\left}{\mathopen{}\mathclose\bgroup\originalleft}
\renewcommand{\right}{\aftergroup\egroup\originalright}

\newcommand{\dens}{\rho}
\newcommand{\prof}{v}

\newcommand{\densestgreedy}{\mathtt{DensestGreedy}}

\newcommand{\Th}{^{\mathrm{th}}}

\newcommand{\abs}[1]{\left|#1\right|}
\newcommand{\Size}[2]{#1^{\left[#2\right]} }
\newcommand{\Goods}[2]{#1^{\left(#2\right)} }

\newcommand{\Y}{Z}
\newcommand{\tauprime}{\widehat{\tau}}

\newcommand{\structlb}{T}
\newcommand{\structub}{\widehat{\structlb}}

\newcommand{\efcount}{{\sf EFCount}}

\DeclareMathOperator*{\argmax}{arg\,max}
\DeclareMathOperator*{\argmin}{arg\,min}

\usepackage[capitalize]{cleveref}
\usepackage{graphicx}
\usepackage{svg}
\usepackage{float}
\usepackage{verbatim}

\usepackage{thm-restate}

\newcommand{\EFone}{{\sf EF1}}
\newcommand{\EFtwo}{{\sf EF2}}
\newcommand{\EFk}{{\sf EFk}}

\title{\bfseries Finding Fair Allocations under Budget Constraints}
\author{Siddharth Barman\thanks{Indian Institute of Science. {\tt barman@iisc.ac.in}} \quad Arindam Khan\thanks{Indian Institute of Science. {\tt arindamkhan@iisc.ac.in}} \quad Sudarshan Shyam\thanks{Aarhus University. {\tt sudarshans.iitkgp@gmail.com}} \quad K.V.N.~Sreenivas\thanks{Indian Institute of Science. {\tt venkatanaga@iisc.ac.in}}}
\date{\empty}
\begin{document}
\maketitle

\begin{abstract}
We study the fair allocation of indivisible goods among agents with identical, additive valuations but individual budget constraints. Here, the indivisible goods---each with a specific size and value---need to be allocated such that the bundle assigned to each agent is of total size at most the agent's budget. Since envy-free allocations do not necessarily exist in the indivisible goods context, compelling relaxations---in particular, the notion of {\em envy-freeness up to $k$ goods} (\EFk)---have received significant attention in recent years. In an $\EFk$ allocation, each agent prefers its own bundle over that of any other agent, up to the removal of $k$ goods, and the agents have similarly bounded envy against the charity (which corresponds to the set of all unallocated goods). Recently, Wu et al.~(2021) showed that an allocation that satisfies the budget constraints and maximizes the Nash social welfare is $1/4$-approximately $\EFone$. However, the computation (or even existence) of exact $\EFk$ allocations remained an intriguing open problem.

We make notable progress towards this by proposing a simple, greedy, polynomial-time algorithm that computes $\EFtwo$ allocations under budget constraints. Our algorithmic result  implies the universal existence of $\EFtwo$ allocations in this fair division context. The analysis of the algorithm exploits intricate structural properties of envy-freeness. Interestingly, the same algorithm also provides $\EFone$ guarantees for important special cases. Specifically, we settle the existence of $\EFone$ allocations for instances in which: (i) the value of each good is proportional to its size, (ii) all goods have the same size, or (iii) all the goods have the same value. Our $\EFtwo$ result extends to the setting wherein the goods' sizes are agent specific.
\end{abstract}

\section{Introduction}
Discrete fair division is an actively growing field of research at the interface of computer science, mathematical economics, and multi-agent systems \cite{handbook2016,aziz2022algorithmic,amanatidis2022fair}. This study is motivated, in large part, by resource-allocation settings in which the underlying resources have to be assigned integrally and cannot be fractionally divided among the agents. Notable examples of such settings include fair allocation of courses \cite{Budish2011TheCA,Othman2010}, public housing units \cite{deng2013story}, and inheritance \cite{Spliddit}. 

A distinguishing feature of discrete fair division is its development of fairness notions that are applicable in the context of indivisible goods. A focus on relaxations is necessitated by the fact that existential guarantees, under classic fairness notions, are scarce in the context of indivisible goods. In particular, the fundamental fairness criterion of envy-freeness---which requires that each agent values the bundle assigned to her over that of any other agent---cannot be guaranteed in the indivisible-goods setting; consider the simple example of a single indivisible good and multiple agents. Interestingy, such pathology can be addressed by considering a natural relaxation: prior works have shown that, among agents with monotone valuations, there necessarily exists an allocation in which envy towards any agent can be resolved by the removal of a good \cite{LMMS042,Budish2011TheCA}.

More generally, recent research in discrete fair division has addressed existential and algorithmic questions related to the notion of envy-freeness up to $k$ goods (\EFk); see, e.g., the survey \cite{suksompong2021constraints} and references therein. In an $\EFk$ allocation, each agent prefers its own bundle over that of any other agent, up to the removal of $k$ goods from the other agent's bundle. A mature understanding has been developed in recent years specifically for allocations that are envy-free up to one good ($\EFone$), e.g., it is known that, under additive valuations, Pareto efficiency can be achieved in conjunction with $\EFone$ \cite{caragiannis2019unreasonable,barman2018finding}. 

However, most works on $\EFone$, and further relaxations, assume that all possible assignments of the indivisible goods (among the agents) are feasible. On the other hand, combinatorial constraints are an unavoidable requirement in many resource-allocation settings. As an illustrative example to highlight the significance of constraints in discrete fair division settings, consider a curator tasked with fairly partitioning  artwork among different museums (i.e., among different agents).\footnote{This stylized example is adapted from \cite{gourves2014near}} Each artifact (indivisible good) has an associated value and a space requirement (depending on its size). Note that the artifacts assigned to a particular museum must fit within its premises and, hence, in this setting not all allocations are feasible. Indeed, here the curator needs to identify a partition of the artifacts that is not only fair but also feasible with respect to the museums' space constraints. The current work addresses an abstraction of this problem. 

A bit more formally, we study the fair allocation of $m$ indivisible goods among $n$ agents with identical, additive valuations but individual budget constraints. Here, each good $g \in [m]$ has a size $s(g) \in \mathbb{Q}_+$ and value $v(g) \in \mathbb{Q}_+$, and the goods need to be partitioned such that bundle assigned to each agent $a \in [n]$ is of total size at most the agent's budget $B_a \in \mathbb{Q}_+$. Note that in this setting with constraints, one might not be able to assign all the $m$ goods among the $n$ agents. Specifically, consider a case wherein the total size of all the goods $s([m]) > \sum_{a=1}^n B_a$. To account for goods that might remain unallocated, we utilize the construct of \emph{charity}. This idea has been used in multiple prior works; see, e.g., \cite{wu2021budget,chaudhury2021little}. The subset of goods that are not assigned to any of the $n$ agents are, by default, given to the charity.

In this framework, we consider envy-freeness up to $k$ goods (\EFk) while respecting the budget constraints. Recall that in the current model, the agents have identical, additive valuations, i.e., for any agent $a \in [n]$ the value of any subset of goods $S \subseteq [m]$ is the sum of values of the goods in it, $v(S) \coloneqq \sum_{g \in S} v(g)$. Also, we say that---for an agent $a \in [n]$ with assigned bundle $A_a \subseteq [m]$---$\EFk$ holds against a subset $F$ iff the value of the assigned bundle, $v(A_a)$, is at least as much as the value of $F$, up to the removal of $k$ goods from $F$. Overall, an allocation $A_1, A_2, \ldots A_n$ (in which agent $a \in [n]$ receives subset $A_a$) is deemed to be $\EFk$ iff, for each pair of agents $a, b \in [n]$ and every subset $F \subseteq A_b$ of size at most $B_a$, the $\EFk$ guarantee holds for agent $a$ against subset $F$. That is, while evaluating envy from agent $a$ towards agent $b$, we consider, within $A_b$, all subsets of size at most $a$'s budget $B_a$.

The fair-division model with budget constraints was proposed by Wu et al.~\cite{wu2021budget}. Addressing agents with distinct, additive valuations, Wu et al.~\cite{wu2021budget} show that an allocation that satisfies the budget constraints and maximizes the Nash Social Welfare is $1/4$-approximately $\EFone$. A manuscript by Gan et al.~\cite{gan2021approximately} improves this guarantee to $1/2$-approximately $\EFone$ for agents with identical, additive valuations. In addition, Gan et al.~\cite{gan2021approximately} show that if all the agents have the same budget and in the case of two agents, an $\EFone$ allocation can be computed efficiently. However, for a general number of agents with distinct budgets, the computation (or even existence) of \emph{exact} $\EFk$ allocations remained an intriguing open problem.  \\

\noindent 
{\bf Our Results and Techniques.}
We make notable progress towards this open question by proposing a simple, greedy, polynomial-time algorithm that computes $\EFtwo$ allocations under budget constraints (Theorem \ref{thm:densestgreedy-ef2}). Our algorithmic result  implies the universal existence of $\EFtwo$ allocations in this fair division context. The same algorithm also provides $\EFone$ guarantees for important special cases. Specifically, we settle the existence of $\EFone$ allocations for instances in which: (i) the value of each good is proportional to its size, (ii) all goods have the same size, or (iii) all the goods have the same value; see Theorems \ref{thm:unit-density-ef1}, \ref{thm:unit-size-ef1}, and \ref{theorem:cardinality}. That is, we prove that, if the densities, values, or sizes of the goods are homogeneous, then an $\EFone$ allocation is guaranteed to exist. Furthermore, our $\EFtwo$ result even extends to the setting wherein the goods' sizes are agent specific (see Theorem \ref{thm:densestgreedy-ef2_gap} in Appendix \ref{appendix:mini-GAP}).

Our algorithm (Algorithm \ref{algo:dense}) allocates goods in decreasing order of density,\footnote{The density of a good $g$ is defined to be its value-by-size ratio, $v(g)/s(g)$.} while maintaining the budget constraints. It is relevant to note that, while the design of the algorithm is simple, its analysis rests on  intricate structural properties of envy-freeness under budget constraints. We obtain the $\EFtwo$ and $\EFone$ guarantees using ideas that are notably different for the ones used in the unconstrained settings; in particular, different from analysis of the envy-cycle-elimination method or the round-robin algorithm \cite{aziz2022algorithmic}. 

Complementing the robustness of the algorithm, we also provide an example that shows that the greedy algorithm might not find an $\EFone$ allocation, i.e., the $\EFtwo$ guarantee is tight (Section \ref{section:tightness}).  \\

\noindent
{\bf The Knapsack Problem.}
The budget constraints, as considered in this work, are the defining feature of the classic knapsack problem. The knapsack problem and its numerous variants have been extensively studied in combinatorial optimization, approximation and online algorithms \cite{martello1990knapsack, knapbook, albers2021improved}. The knapsack problem finds many applications in practice \cite{skiena1999interested}. Recall that the objective in the knapsack problem is to find a subset with maximum possible value, subject to a single budget constraint. That is, the goal in the standard knapsack problem is utilitarian and not concerned with fairness.   

Algorithmic aspects of the special cases considered in the current paper have been addressed in prior works: (i) knapsack instances in which the value of each good is proportional to its size are known as proportional instances~\cite{cygan2016online} or subset-sum instances~\cite{pisinger2005hard}, (ii) instances where all the goods have the same value are referred to as cardinality \cite{galvez2021approximating, Galvez00RW21, K0001MSW21} or unit \cite{cygan2016online} instances. In addition, we also study the special case wherein all the goods have the same size. 

Proportional and cardinality versions of the knapsack problem are known to be technically challenging by themselves. In particular, in the context of online algorithms it is known that there does not exist a deterministic algorithm with bounded competitive ratio for these two versions \cite{lueker1998average, marchetti1995stochastic}. 

The knapsack problem has also been studied from the perspective of group fairness \cite{Patel0L21} and fairness in aggregating voters' preferences \cite{FluschnikSTW19}. In these works, there is only one knapsack and the single, selected subset of goods induces (possibly distinct) valuations among the agents. By contrast, the current work addresses multiple knapsacks, one for each agent. \\

\noindent{\bf Generalized Assignment Problem (GAP).} We also address instances in which the goods' sizes are agent specific. While such instances constitute a generalization of the formulation considered in the rest of the paper, they themselves are a special case of GAP \cite{shmoys1993approximation}. GAP consists of goods both whose sizes and values are agent specific. For GAP, it is known that value maximization is {\rm APX}-hard. In fact, even with common values and agent-specific sizes, the value-maximization objective does \emph{not} admit a polynomial-time approximation scheme \cite{chekuri2005polynomial}. \\  

\noindent
{\bf Additional Related Work.}
As mentioned previously, $\EFk$ allocations have been studied in various discrete fair division contexts~\cite{suksompong2021constraints}. In particular, Bil{\`o} et al.~\cite{bilo2022almost} consider settings in which the indivisible goods correspond to vertices of a given graph $G$ and each agent must receive a connected subgraph of $G$. It is shown in \cite{bilo2022almost} that if the graph $G$ is a path, then, under the connectivity constraint, an $\EFtwo$ allocation is guaranteed to exist. Under connectivity constraints imposed by general graphs $G$, Bei et al.~\cite{bei2022price} characterize the smallest $k$ for which an $\EFk$ allocation necessarily exists among two agents (i.e., this result addresses the $n=2$ case). We also note that exact $\EFk$ guarantees are incomparable with multiplicative approximations, as obtained in \cite{wu2021budget}.  

The current work focusses on settings in which the agents have an identical (additive) valuation over the goods. We note that, in the context of budget constraints, identical valuations already provide a technically-rich model. Fair division with identical valuations has been studied in multiple prior works; see, e.g.,~\cite{plaut2020almost,barman2021uniform}. Indeed, in many application domains each agent's cardinal preference corresponds to the monetary worth of the goods and, hence, in such setups the agents share a common valuation.

\section{Notation and Preliminaries}
We study the problem of fairly allocating a set of indivisible goods $[m] = \{1,2, \dots, m\}$ among a set of agents $[n] = \{1,2, \dots, n \}$ with budget constraints.
In the setup, every good $g \in [m]$ has a size $s(g) \in \mathbb{Q}_+$ and a value $\prof(g) \in \mathbb{Q}_+$. The density of any good $g \in [m]$ will be denoted as $\dens(g) \coloneqq \prof(g)/s(g)$. Furthermore, every agent $a \in [n]$ has an associated budget $B_a \in \mathbb{Q}_+$ that specifies an upper bound on the cumulative size of the set of goods that agent $a$ can receive. We conform to the framework wherein the valuations and sizes of the goods are additive; in particular, for any subset of goods $S \subseteq [m]$, we write the value $\prof(S) \coloneqq  \sum_{g \in S} \prof(g)$ and the size $s(S) \coloneqq  \sum_{g \in S} s(g)$. Hence, in this setup, a subset $S \subseteq [m]$ can be assigned to agent $a \in [n]$ only if $s(S) \leq B_a$, and the subset has value $v(S)$ for the agent. An instance of the fair division problem with budget constraints is specified as a tuple $\langle [m], [n], \{ v(g) \}_{g\in [m]}, \{s(g)\}_{g \in [m]}, \{B_a\}_{a\in [n]} \rangle$. 

Note that in fair division settings with constraints, one might not be able to assign all the $m$ goods among the $n$ agents. Specifically, consider a setting wherein $s([m]) > \sum_{a=1}^n B_a$. To account for goods that might remain unallocated, we utilize the construct of \emph{charity}. 
The goods that are not assigned to any of the $n$ agents are, by default, given to the charity. 


An allocation $\mathcal{A} = (A_1, A_2, \ldots, A_n)$ refers to a tuple of disjoint sets of goods, i.e., for every $a\in[n]$, $A_a\subseteq G$
and for all $a,b\in[n]$ such that $a\ne b$, $A_a\cap A_b=\emptyset$. Here $A_a$ indicates the set of goods allocated to agent $a$.
Throughout, we will maintain allocations $\mathcal{A} = (A_1, A_2, \ldots, A_n)$ that are feasible, i.e., satisfy the budget constraints of all the agents, $s(A_a) \leq B_a$ for every agent $a \in [n]$. 
As mentioned above, the set of remaining goods, $[m]\setminus (A_1\cup A_2\cup \dots\cup A_n)$, will be assigned to the charity.

Next, we define the fairness notions studied in this work. Consider an allocation $\mathcal{A} = (A_1, A_2, \ldots, A_n)$.
An agent $a\in[n]$ is said to be \emph{envy-free up to one good} ($\EFone$) towards agent $b\in[n]$
iff for every subset $F \subseteq A_b$, with $s(F) \leq B_a$ (and $|F| \geq 1$), there exists a 
good $f \in F$ such that $v(A_a) \geq v(F \setminus \{f\})$. Further, an agent $a\in[n]$ is said to be $\EFone$ towards the charity
iff for every subset $F \subseteq [m]\setminus \cup_{a=1}^n A_a$, with $s(F) \leq B_a$ (and $|F| \geq 1$), there exists a 
good $f \in F$ such that $v(A_a) \geq v(F \setminus \{f\})$.
The allocation $\mathcal A$ is said to be $\EFone$ iff every agent $a\in[n]$ is $\EFone$ towards every other agent $b\in[n]$ and the charity.

Analogously, we define $\EFtwo$:

\begin{definition}[$\EFtwo$]
\label{defn:eftwo}
Let $\mathcal{A} = (A_1, A_2, \ldots, A_n)$ be an arbitrary allocation.
An agent $a\in[n]$ is said to be \emph{envy-free up to two goods} ($\EFtwo$) towards agent $b\in[n]$
iff for every subset $F \subseteq A_b$, with $s(F) \leq B_a$ (and $|F| \geq 2$), there exist
goods $f_1,f_2 \in F$ such that $v(A_a) \geq v(F \setminus \{f_1,f_2\})$. Further, an agent $a\in[n]$ is said to be $\EFtwo$ towards the charity
iff for every subset $F \subseteq [m]\setminus \cup_{a=1}^n A_a$, with $s(F) \leq B_a$ (and $|F| \geq2$), there exist goods 
$f_1,f_2 \in F$ such that $v(A_a) \geq v(F \setminus \{f_1,f_2\})$.
The allocation $\mathcal A$ is said to be $\EFtwo$ iff every agent $a\in[n]$ is $\EFtwo$ towards every other agent $b\in[n]$ and the charity.
\end{definition}

Throughout, we will assume that the goods have distinct densities -- this assumption holds without loss of generality and can be enforced by perturbing the densities (and appropriately the values) by sufficiently small amounts (see Appendix \ref{appendix:distinct-densities}). The assumption ensures that, in any nonempty subset $S \subseteq [m]$, the good with the maximum density ${\argmax_{g \in S} \ \dens(g)}$ is uniquely defined. Also, indexing the goods, in any subset $S =\{g_1, g_2, \ldots, g_k \}$, in decreasing order of density results in a unique ordering with $\dens(g_1) > \dens(g_2) > \ldots > \dens(g_k)$. For any subset $S \subseteq [m]$ and good $g \in [m]$, we will use the shorthands $S + g \coloneqq S \cup \{g \}$ and $S - g \coloneqq S \setminus \{g\}$.

\section{The Density Greedy Algorithm}
\label{section3:densestgreedy-ef2}
This section develops a greedy algorithm (Algorithm \ref{algo:dense} - $\densestgreedy$) that allocates goods in decreasing order of density, while maintaining the budget constraints. We will prove that the algorithm achieves  $\EFtwo$ for general budget-constrained instances and $\EFone$ for multiple special cases. 

\begin{theorem}
\label{thm:densestgreedy-ef2}
For any given fair division instance with budget constraints $\langle [m], [n], \{ v(g) \}_{g\in [m]}, \{s(g)\}_{g \in [m]}, \{B_a\}_{a\in [n]} \rangle$, Algorithm \ref{algo:dense} ($\densestgreedy$) computes an $\EFtwo$ allocation in polynomial time.
\end{theorem}

\begin{algorithm}
\caption{$\densestgreedy$ -- Given instance $\langle [m], [n], \{ v(g) \}_{g}, \{s(g)\}_{g}, \{B_a\}_{a} \rangle$,  
allocate the goods $[m]$ among agents $[n]$ (the unassigned goods go to charity).}  \label{algo:dense}
\begin{algorithmic}[1]
\STATE Initialize allocation $(A_1, \dots, A_{n}) \leftarrow (\emptyset, \ldots, \emptyset)$. Also, define the set of active agents $N \coloneqq [n]$ and the set of unallocated goods $G \coloneqq [m]$. 
\WHILE{$G \neq \emptyset$ and $N \neq \emptyset$}
\STATE Select arbitrarily a minimum-valued agent $a \in N$, i.e., $a = {\displaystyle \argmin_{b \in N} \prof(A_b)}$. \label{line3}
\IF{for all goods $g\in G$ we have $s\left( A_a + g \right) > B_a$}
\STATE Set agent $a$ to be inactive, i.e., $N \leftarrow N \setminus \{a\}$.
\ELSE
\STATE \label{line:assign-good} Write $g' = {\displaystyle \argmax_{g \in G :\  s(A_a + g) \leq B_a} \dens(g)}$ 
and update $A_a \leftarrow A_a + g'$ along with $G \leftarrow G - g'$.
\ENDIF
\ENDWHILE
\RETURN $(A_1,A_2,\dots,A_{n})$
\end{algorithmic}
\end{algorithm}
Recall the assumption that the goods have distinct densities and, hence, in Line \ref{line:assign-good} we obtain a unique good $g'$ (among the ones that fit within agent $a$'s available budget). While our goal is to find a fair, integral allocation of the (indivisible) goods, for analytic purposes, we will consider fractional assignment of goods to agents. Towards this, for any scalar $\alpha \in[0,1]$ and good $g \in [m]$, we define $\alpha \cdot g$ to be a new good whose size and value are $\alpha$ times the size and value of good $g$, respectively. With fractional goods, we obtain set difference between subsets $I$ and $J$ by adjusting the fractional amount of each good present in $I$. Formally, for subsets $I = \{g_1, g_2, \ldots, g_k\}$ and $J$, let $\alpha_i$ denote the fraction of the good $g_i$ present in $J$,\footnote{Note that, if $\alpha_i = 0$, then the good $g_i$ is not included in $J$. Complementarily, if $\alpha_i = 1$, then the good $g_i$ is completely included in $J$. Also, $g_i$ itself could be a fractional good.} then $I \setminus J \coloneqq \cup_{i=1}^k \left\{ (1-\alpha_i) \cdot g_i \right\}$.  

We next define key constructs for the analysis. For any subset of goods $S$, we define two density-wise prefix subsets of $S$; in particular, $S^{(i)}$ is the subset of the $i$ densest goods in $S$ and $S^{[B]}$ consists of the densest goods in $S$ of total size $B$. Formally, for any subset of goods $S = \{s_1, s_2, \ldots, s_k \}$, indexed in decreasing order of density, and any index $1 \leq i \leq |S|$, write $S^{(i)} \coloneqq \{s_1, \ldots, s_i \}$.
\begin{definition}[Prefix Subset $\Size{S}{B}$]
\label{definition:prefix-subset}
For any subset of goods $S = \{g_1, g_2, \ldots, g_k \}$, indexed in decreasing order of density, and for any nonnegative threshold $B \le s(S)$, let $P=\{g_1, \ldots, g_{\ell-1} \}$ be the (cardinality-wise) largest prefix of $S$ such that $s(P) \leq B$. Then, we define $\Size{S}{B} \coloneqq P \cup \left\{ \alpha \cdot g_{\ell} \right\}$, where $\alpha = \frac{B-s(P)}{s(g_{\ell})}$.  
\end{definition}
If the threshold $B \geq s(S)$, then we simply set $\Size{S}{B}=S$. Note that in $\Size{S}{B}$ at most one good is fractional and, for $B \leq s(S)$, the size of $\Size{S}{B}$ is exactly equal to $B$. It is also relevant to observe that, if $A_a$ is the subset of goods assigned to agent $a \in [n]$ at the end of Algorithm \ref{algo:dense}, then $A^{(i)}_a$ is in fact the set of the first $i$ goods assigned to agent $a$ in the algorithm; recall that the algorithm assigns the goods in decreasing order of density. The following two propositions provide useful properties of Algorithm \ref{algo:dense} and are based on the algorithm's selection criteria. The proofs of these propositions appear in Appendix \ref{appendix:proposition-proofs}.

\begin{restatable}{proposition}{PropGreedyOne}
\label{greedyproperty}
Let $X=\{g_1, g_2, \ldots, g_k\}$  denote the set of goods assigned to an agent $a \in [n]$ (i.e., $X = A_a$) and $Y=\{h_1, h_2, \ldots, h_\ell\}$ 
be the set of goods assigned to one of the agents $b \in [n]$, {or} to the charity (i.e., $Y = A_b$ or $Y = [m] \backslash \cup_{i=1}^n A_i$) at the end of Algorithm \ref{algo:dense}. Further, let the goods in the sets $X$ and $Y$ be indexed in decreasing order of density.  For indices $i < |X|$ and $j < |Y|$, suppose $\prof\left(X^{(i)} \right) < \prof\left(Y^{(j)}\right)$ and $s\left(X^{(i)} + h_{j+1}\right) \leq B_a$. Then, $\dens(g_{i+1}) > \dens(h_{j+1})$. 
\end{restatable}

\begin{restatable}{proposition}{PropGreedyTwo}
\label{greedyproperty-full}
Let $X=\{g_1, g_2, \ldots, g_k\}$  denote the set of goods assigned to an agent $a \in [n]$ and $Y=\{h_1, h_2, \ldots, h_\ell\}$ 
be the set of goods assigned to one of the agents $b \in [n]$, {or} to the charity, at the end of Algorithm \ref{algo:dense}. 
Further, let the goods in the sets $X$ and $Y$ be indexed in decreasing order of density. 
If, for any index $j < \abs{X}$, the size $s\left(X+h_{j+1}\right)\le B_a$, then we have $\prof\left(X\right)\geq \prof\left(Y^{(j)}\right)$.
\end{restatable}

We define the function $\efcount(\cdot)$ to capture envy count, i.e., the number of goods that need to be removed in order to achieve envy-freeness. Specifically, 
for any subset of goods $X,Y$, we define $\efcount(X,Y)$ as the minimum number of goods whose removal from $Y$ yields a subset of goods with value at most $\prof(X)$,  
\begin{align}
\efcount(X,Y) \coloneqq \min_{R \subseteq Y: \ \prof({Y \setminus R} ) \leq \prof(X)} \  |R| \label{eqn:efcount-defn}
\end{align}

\subsection{Structural Properties of Envy Counts}
\label{section:structural-props}

This section develops important building blocks for the algorithm's analysis. 
\begin{lemma}\label{lemma:lipschitz}
For any subset of goods $X$ and $Y$ along with any index $i < |Y|$, let $\structlb \coloneqq s\left(Y^{(i)}\right)$ and $\structub \coloneqq s\left(Y^{(i+1)}\right)$. 
Then, $\efcount\left(\Size{X}{\structub}, \Size{Y}{\structub}\right) \leq \efcount\left(\Size{X}{\structlb}, \Size{Y}{\structlb}\right) + 1$.
\end{lemma}
\begin{proof}
Write $c \coloneqq \efcount\left(\Size{X}{\structlb}, \Size{Y}{\structlb}\right)$. Therefore, by definition, there exists a size-$c$ subset $R \subseteq Y^{[\structlb]}$ with the property that 
$\prof(X^{[\structlb]}) \geq \prof(Y^{[\structlb]} \setminus R)$. Define subset $R' := R \cup \{h_{i+1}\}$,  where $h_{i+1}$ is the good in the set $Y^{(i+1)}\setminus Y^{(i)}$. For this set $R'$ of cardinality $c+1$, we have
\begin{align*}
\prof\left(X^{[\structub]}\right) &\geq \prof\left(X^{[\structlb]}\right) 
     \geq \prof\left(Y^{[\structlb]} \setminus R\right) 
     = \prof\left(Y^{[\structub]} \setminus R'\right).
\end{align*}
This implies $\efcount\left(\Size{X}{\structub}, \Size{Y}{\structub}\right)\leq c + 1$, and the lemma stands proved. 
\end{proof}

The following lemma shows that if $\efcount$ from a subset $X$ to a subset $Y$ is more than two, then we can select prefix subsets of $X$ and $Y$ such that the count becomes exactly equal to two.  
\begin{lemma}\label{existenceof2cut}
Let $X$ and $Y$ be any subsets of goods with the property that $\efcount(X,Y) \geq 2$. Then, there exists an index $t \leq |Y|$ such that, with $\structlb \coloneqq s\left(\Goods{Y}{t}\right)$, we have  $\efcount\left(\Size{X}{\structlb}, \Size{Y}{\structlb}\right) = 2$.
\end{lemma}
\begin{proof}
The lemma essentially follows from a discrete version of the intermediate value theorem. For indices $t \in \{0,1,2,\ldots, |Y|\}$, define the function $h(t) \coloneqq s\left(\Goods{Y} {t}\right)$, i.e., $h(t)$ denotes the size of the $t$ densest goods in $Y$. Extending this function, we consider the envy count at different size thresholds; in particular, write $H(t) \coloneqq \efcount\left(\Size{X}{h(t)}, \Size{Y}{h(t)}\right)$ for each $t \in \{0,1,2,\ldots, |Y|\}$. Note that $H(0) = 0$. We will next show that $(i)$ $H(|Y|) \geq 2$ and $(ii)$ the discrete derivative of $H$ is at most one, i.e., $H(t+1) - H(t) \leq 1$ for all $0 \leq t < |Y|$. These properties of the integer-valued function $H$ imply that there necessarily exists an index $t^*$ such that $H(t^*) = 2$. This index $t^*$  satisfies the lemma.

Therefore, we complete the proof by establishing properties $(i)$ and $(ii)$ for the function $H(\cdot)$. For $(i)$, note that the definition of the prefix subset gives us $\prof(X)\ge \prof\left(\Size{X}{s(Y)}\right)$. Hence, $\efcount\left(\Size X {s(Y)}, Y\right) \geq \efcount(X, Y) \geq 2$; the last inequality follows from the lemma assumption.  Since $h(|Y|) = s(Y)$, we have $H(|Y|) \geq 2$. Property $(ii)$ follows directly from Lemma \ref{lemma:lipschitz}. This completes the proof. 
\end{proof}

The next lemma will be critical in our analysis. At a high level, it asserts that if we have two subsets $X'$ and $\Y'$ with $\efcount(X',\Y') = 2$ and one adds more value into $X'$ than $\Y'$, then the envy count does not increase. 

\begin{lemma}\label{envy_diff}
Given two subsets of goods $X$ and $\Y$ along with two nonnegative size thresholds $\structlb,\structub \in \mathbb{R}_+$ with the properties that  
\begin{itemize}
\item $\efcount\left(\Size{X}{\structlb}, \Size \Y {\structub}\right) = 2$ \ and
\item $\prof\left(X \setminus \Size X {\structlb}\right) \geq  \prof\left(\Y \setminus \Size \Y {\structub}\right)$.
\end{itemize}
Then, $\efcount(X,\Y) \leq \efcount\left(\Size X {\structlb}, \Size \Y {\structub}\right)=2$.
\end{lemma}
\begin{proof}
Given that $\efcount\left(\Size X {\structlb}, \Size \Y {\structub}\right) = 2$, there exist two goods $g_1', g_2' \in \Size \Y {\structub}$
such that $\prof\left(\Size{\Y}{\structub} - g_1' - g_2'\right)  \leq \prof\left(\Size X {\structlb}\right)$. Now, using the definition of the prefix subsets (Definition \ref{definition:prefix-subset}) we get
\begin{align}
	\prof(X) &= \prof\left(\Size X {\structlb}\right) + \prof\left(X \setminus \Size X {\structlb}\right) \nonumber \\
	& \geq \prof\left(\Size \Y {\structub} - g_1' - g_2'\right) + \prof\left(X \setminus \Size X {\structlb}\right) \nonumber \\
	& \geq \prof\left(\Size \Y {\structub} - g_1' - g_2'\right) +  \prof\left(\Y \setminus \Size \Y {\structub}\right) \tag{via lemma assumption} \\
	& = \prof(\Y)  -\left( \prof(g'_1) + \prof(g'_2) \right) \label{ineq:interim}
\end{align}
The definition of the prefix subset $\Size{\Y}{\structub}$ ensures that, corresponding to goods $g'_1, g'_2 \in \Size{\Y}{\structub}$, there exist two goods $g_1, g_2 \in \Y$ such that $v(g_1) + v(g_2) \geq v(g'_1) + v(g'_2)$. This bound and inequality (\ref{ineq:interim}) give us $\prof(X) \geq \prof(\Y - g_1 - g_2)$. This implies $\efcount(X,\Y) \leq 2$ and completes the proof of the lemma. 
\end{proof}
\begin{remark}
Lemma \ref{envy_diff} continues to hold when $\efcount\left(\Size{X}{\structlb}, \Size{\Y}{\structub}\right) = c$, for any integer $c\geq1$.
\end{remark}

\subsection{Proof of Theorem \ref{thm:densestgreedy-ef2}: $\densestgreedy$ achieves EF2}
\label{section:proof-of-main-theorem}

This subsection establishes Theorem \ref{thm:densestgreedy-ef2}. The runtime analysis of the algorithm is direct. Therefore, we focus on proving that $\densestgreedy$ necessarily finds an $\EFtwo$ allocation. Towards this, let $X = \{x_1, x_2, \dots, x_k \}$ be the subset of goods allocated to an agent $a \in [n]$, and  let $Y =  \{y_1, y_2, \dots, y_{\ell} \}$ be the subset of goods allocated to an other agent $b \in [n]$ or to the charity at the end of $\densestgreedy$. The goods in both $X$ and $Y$ are indexed in decreasing order of density. Establishing $\EFtwo$ between the sets of goods $X$ and $Y$ corresponds to showing that, for any subset of goods $\Y \subseteq Y$, with $s(\Y) \leq B_a$, we have $\efcount(X,\Y) \leq 2$. 

Consider any such subset $\Y$ and index its goods in decreasing order of density, $\Y = \{z_1,z_2,\dots,z_{\ell'} \}$. Note that, if $\efcount(X, \Y) \leq 1$, we already have the $\EFtwo$ guarantee. Therefore, in the remainder of the proof we address the case wherein $\efcount(X, \Y) \geq 2$. We will in fact show that this inequality cannot be strict, i.e., it must hold that the envy count is at most $2$ and, hence, we will obtain the $\EFtwo$ guarantee. Our proof relies on carefully identifying certain prefix subsets, showing that they satisfy relevant properties, and finally invoking Lemma \ref{envy_diff}.  

We start by considering function $h(i)$ which denotes the size of the $i$ densest goods in set $Z$, i.e., $h(i) \coloneqq s(\Goods{Z}{i})$ for $i \in \{0,1,2,\ldots, |Z|\}$. Furthermore, define index
\begin{align}
\label{eq:ef2}
t \coloneqq \min\left\{ i : \efcount\left(\Size{X}{h(i)}, \Goods{\Y}{i}\right) = 2\right\}
\end{align}

\begin{figure}[h]
\centering
\includegraphics[scale=0.9]{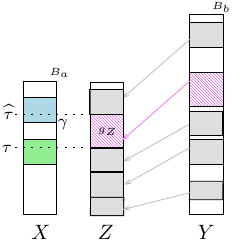} 
\caption{Figure illustrating size thresholds $\tau$, $\tauprime$, and the good $g_{\Y}$.} \label{figure:cut-offs}
\end{figure}
Existence of such an index $t \ge 2$ follows from Lemma \ref{existenceof2cut}. Also, note that $\Goods{\Y}{i} = \Size{\Y}{h(i)}$. 
We will denote the $t$th good in $\Y$ by $g_{\Y}$, i.e., $g_{\Y} = z_{t}$. In addition, using $t$ we define the following two size thresholds (see Figure \ref{figure:cut-offs}):
\begin{align}
\tau \coloneqq s\left(\Goods \Y {t-1}\right) \quad \text{ and } \quad \tauprime \coloneqq s\left(\Goods \Y t\right)\label{eq:defn-tau}
\end{align}
That is, $\tau = h(t-1)$ and $\tauprime = h(t)$. Now, from Lemma \ref{lemma:lipschitz} and the definition of $t$ (equation (\ref{eq:ef2})) we can infer that $\efcount\left(\Size X {\tau}, \Size \Y {\tau}\right) \ge 1$. Furthermore, using the minimality of $t$ we get $\efcount\left(\Size X {\tau}, \Size \Y {\tau}\right) < 2$.
Hence, 
\begin{align}
\efcount\left(\Size X {\tau}, \Size \Y {\tau}\right) = 1 \label{eqn:xtau-ztau}
\end{align}

We will establish two properties for the sets $X$ and $Z$ under consideration and use them to invoke Lemma \ref{envy_diff}. Specifically, in Lemma \ref{small_ef2} we will show that  $\efcount\left(\Size X {\tau}, \Size \Y \tauprime\right) = 2$ and in Lemma \ref{lemma} we prove $\prof\left(X \setminus \Size X \tau\right) \geq \prof\left(\Y \setminus \Size \Y \tauprime\right)$. These are exactly the two properties required to apply Lemma \ref{envy_diff} with $  \structlb = \tau$ and $ \structub = \tauprime$. 

\begin{lemma}\label{small_ef2}
$\efcount\left(\Size X {\tau}, \Size \Y \tauprime\right) = 2$.
\end{lemma}
\begin{proof}
Since $\efcount\left(\Size X \tau, \Size \Y \tau\right) = 1$ (see equation (\ref{eqn:xtau-ztau})), there exists a good $g_1 \in \Size \Y \tau$ such that $\prof\left(\Size{X}{\tau}\right) \geq \prof\left(\Size{\Y}{\tau} - g_1\right)$.  
Also, by definition, we have $ \Size{\Y}{\tauprime} = \Size{\Y}{\tau} \cup \{ g_{\Y} \}$. Hence, the previous inequality reduces to $\prof\left(\Size{X}{\tau}\right) \geq \prof\left(\Size \Y \tauprime -  g_{\Y} - g_1\right)$. That is, removing $g_1$ and $g_{\Y}$ from $\Size{\Y}{\tauprime}$ gives us a set with value at most that of $\Size{X}{\tau}$. Therefore, we have $\efcount\left(\Size X {\tau}, \Size \Y \tauprime\right) = 2$. The lemma stands proved. 
\end{proof}

We define $\gamma$ as the size of the goods in $X$ that are at least as dense as $g_{\Y}$, i.e.,  
\begin{align}
\gamma \coloneqq \sum_{g \in X: \dens(g) \geq \dens(g_{\Y})} s(g) \label{eq:def-g}
\end{align}

We will establish bounds considering $\gamma$ and use them to prove Lemma \ref{lemma} below. 

\begin{claim} \label{lemma5}
It holds that $\gamma \leq \tauprime$ and $\prof\left(\Size X \gamma\right) < \prof\left(\Size \Y \tau\right)$. 
\end{claim}
\begin{proof}
We will first establish the stated upper bound on $\gamma$. Assume, towards a contradiction, that $\gamma > \tauprime$. By definition of $\gamma$, we have that all the goods in $\Size{X}{\gamma}$ have density at least $\rho(g_{\Y})$. Now, given that $\gamma > \tauprime$, we get that the density of each good in $\Size{X}{\tauprime}$ is at least $\rho(g_{\Y})$. In particular, all the goods in the set $\Size{X}{\tauprime} \setminus \Size{X}{\tau}$ are at least as dense as $g_{\Y}$. Hence, $\prof\left(\Size X \tauprime \setminus \Size X \tau\right) \geq \prof\left(\Size \Y \tauprime \setminus \Size \Y \tau\right) = \prof(g_{\Y})$. This inequality and equation (\ref{eqn:xtau-ztau}) give us $\efcount(\Size{X}{\tauprime}, \Size{\Y}{\tauprime}) \leq 1$; see Lemma \ref{envy_diff}. This bound, however, contradicts the definition of $t$ (and, correspondingly, $\tauprime$) as specified in equation (\ref{eq:ef2}). Hence, by way of contradiction, we obtain the desired upper bound $\gamma \leq \tauprime$.  

Next, we prove the second inequality from the claim. For a contradiction, assume that $\prof\left(\Size X \gamma\right) \geq \prof\left(\Size \Y \tau\right)$. Since $\gamma \leq \tauprime$, we further get $\prof\left(\Size{X}{\tauprime}\right) \geq \prof\left(\Size \Y \tau\right) =  \prof\left(\Size \Y \tauprime - g_{\Y}\right)$. That is, $\efcount\left(\Size X \tauprime, \Size \Y \tauprime\right) \leq 1$. This envy count contradicts the definition of $t$ (and, correspondingly, $\tauprime$); see equation (\ref{eq:ef2}). Therefore, we obtain the second part of the claim. 
\end{proof}

We will now prove Lemma \ref{lemma}.

\begin{lemma}\label{lemma}
$\prof\left(X \setminus \Size X \tau\right) \geq \prof\left(\Y \setminus \Size \Y \tauprime\right)$.
\end{lemma}
\begin{proof}
Since $\Y \subseteq Y$, the good $g_{\Y}$ appears in the subset $Y$. Recall that the goods in the subsets $Z$ and $Y = \{y_1, y_2, \ldots, y_\ell \}$ are indexed in order of decreasing density. Write $t' \in [|Y|]$ to denote the index of $g_{\Y}$ in $Y$ (i.e., $g_{\Y} = y_{t'}$). Claim \ref{lemma5} gives us 
$\prof\left(\Size X \gamma\right) < \prof\left(\Size \Y \tau\right) = \sum_{i=1}^{t-1} \prof(z_i)\le\sum_{i=1}^{t'-1} \prof(y_i)$. That is, $\prof\left(\Size X \gamma\right) < \prof\left(\Goods Y {t'-1} \right)$.  Also, by definition of $\gamma$ (equation \ref{eq:def-g}), we have that the goods in $X \setminus {\Size{X}{\gamma}}$ (if any) have density less than $\rho(g_{\Y})$. These observations and Proposition \ref{greedyproperty} imply that including $g_{\Y}$ in ${\Size{X}{\gamma}}$ must violate agent $a$'s budget $B_a$, i.e., it must be the case that  
\begin{align}
	\gamma + s(g_{\Y}) > B_a\label{eq:gz-doesnt-fit}
\end{align}

Using inequality (\ref{eq:gz-doesnt-fit}), we will prove that $\prof\left(\Size X \gamma \setminus \Size X \tau\right) \geq \prof\left(\Y \setminus \Size \Y \tauprime\right)$. This bound directly implies the lemma, since $\Size{X}{\gamma} \subseteq X$. In particular, the size of the concerned set satisfies 
\begin{align}
s\left(\Size X \gamma \setminus \Size X \tau\right) & = \gamma - \tau  \nonumber \\
&  = \gamma - \tauprime + s(g_{\Y}) \tag{$\tauprime - s(g_{\Y}) = \tau$} \\
& > B_a - \tauprime \tag{via inequality (\ref{eq:gz-doesnt-fit})}\\
& \geq s\left(\Y \setminus \Size \Y \tauprime\right) \label{ineq:sizelb}
\end{align}
The last inequality follows from the facts that $s(\Y) \leq B_a$ and $s( \Size \Y \tauprime) = \tauprime$. 
Furthermore, by definition of $\gamma$, we have that every good $g \in \Size X \gamma \setminus \Size X \tau$ has density $\dens(g) \geq \dens(g_{\Y})$. In addition, for every good $g' \in \Y \setminus \Size \Y \tauprime$, the density $\dens(g') \leq \dens(g_{\Y})$. These bounds on the densities and the sizes of the subsets ${\Size{X}{\gamma}} \setminus \Size X \tau$ and $\Y \setminus \Size \Y \tauprime$ give us $\prof\left( {\Size{X}{\gamma}} \setminus \Size X \tau\right) \geq \prof\left(\Y \setminus \Size \Y \tauprime\right)$. As mentioned previously, this inequality and the containment $\Size{X}{\gamma} \subseteq X$ imply $\prof\left(X \setminus \Size X \tau\right) \geq \prof\left(\Y \setminus \Size \Y \tauprime\right)$. The lemma stands proved. 
\end{proof}

Overall, Lemma \ref{small_ef2} gives us $\efcount\left(\Size X {\tau}, \Size \Y \tauprime\right) = 2$. In addition, via Lemma \ref{lemma}, we have $\prof\left(X \setminus \Size X \tau\right) \geq \prof\left(\Y \setminus \Size \Y \tauprime\right)$. Therefore, applying Lemma \ref{envy_diff}, we conclude that $\efcount(X, \Y) \leq 2$. This establishes the desired $\EFtwo$ guarantee for the allocation computed by Algorithm \ref{algo:dense} and completes the proof of Theorem \ref{thm:densestgreedy-ef2}.

\subsection{Fair Division in Proportional Instances}
\label{section:efone-proportional-instances}

This section shows that, if all the goods have the same density, then an $\EFone$ allocation can be computed in polynomial time. 
While in the rest of the paper we assume that the goods have distinct densities, in this section we in fact address goods with exactly the same densities. We address this technical difference, by simply including any consistent tie breaking rule in Algorithm \ref{algo:dense}. That is, for proportional instances, the $\densestgreedy$ algorithm applies a tie breaking rule (e.g., lowest index first) while selecting among the unallocated goods in Line \ref{line:assign-good}. With this minor modification, all the previously established results (specifically, Lemma \ref{lemma}) continue to hold for proportional instances. Next, we establish the $\EFone$ guarantee for proportional instances.  

\begin{theorem}
\label{thm:unit-density-ef1}
For any given budget-constrained fair division instance $\langle [m], [n], \{ v(g) \}_{g}, \{s(g)\}_{g}, \{B_a\}_{a} \rangle$ in which all the goods have the same density (i.e., ${v(g)}/{s(g)} = {v(g')}/{s(g')}$ for all $g, g' \in [m]$), Algorithm \ref{algo:dense} ($\densestgreedy$) computes an $\EFone$ allocation in polynomial time.
\end{theorem}
\begin{proof}
We use the constructs defined in Sections \ref{section:structural-props} and \ref{section:proof-of-main-theorem}. In particular, let $X$ be the set of goods allocated to an agent $a \in [n]$ and let $Y$ be the set of goods allocated to an agent $b \in [n]$, or to the charity at the end of Algorithm \ref{algo:dense}. Consider any subset $\Y \subseteq Y$, with size $s(\Y) \leq B_a$ (and $|\Y| \geq 2$). By way of contradiction,
we will show that $\efcount(X,\Y) \leq 1$ and, hence, obtain the stated $\EFone$ guarantee.

Assume, towards a contradiction, that $\efcount(X,\Y) \geq 2$. In such a case, the constructs (specifically, $t$, $\tau$, $\tauprime$, and the good $g_{\Y}$) considered in Sections \ref{section:structural-props} and \ref{section:proof-of-main-theorem} are well-defined. Using the previously-established properties of these constructs, we will show that there necessarily exists
a good $g \in \Y$ such that $\prof(X) \geq \prof(\Y - g)$. Hence, by contradiction, we will get that $\efcount(X, \Y) <2$, i.e., $\EFone$ holds between $X$ and $\Y$. 

We first note that the size of the set $X$ is at least $\tau$. This follows from inequality (\ref{ineq:sizelb}), which gives us $s\left(\Size{X}{\gamma} \setminus \Size{X}{\tau}\right)  > 0$ and, hence, we have  $s(X) - \tau \geq s\left(\Size{X}{\gamma}\right) - \tau > 0$. This lower bound on the size of $X$ implies that the prefix subset $\Size{X}{\tau}$ has size exactly equal to $\tau$. 
In addition, $s(\Size{\Y}{\tau}) = \tau$. Now, given that all the goods have the same density, we obtain
\begin{align}
v(\Size{X}{\tau}) = v(\Size{\Y}{\tau}) \label{ineq:valxtau-ztau}
\end{align}
Therefore, 
\begin{align*}
\prof(X) &= \prof\left(\Size X \tau\right) + \prof\left(X \backslash \Size X \tau\right) \\
         &= \  \prof\left(\Size \Y \tau\right) + \prof\left(X \backslash \Size X \tau\right) \tag{via equation (\ref{ineq:valxtau-ztau})}\\
         &\geq  \prof\left(\Size \Y \tau\right) + \prof\left(\Y \backslash \Size \Y \tauprime \right) \tag{via Lemma \ref{lemma}} \\
        &=  \prof\left(\Size \Y \tau\right) + \prof\left(\Y \backslash \Size \Y \tau\right) - \prof(g_{\Y}) \tag{since $\Size{\Y}{\tauprime} \setminus \Size{\Y}{\tau} = \{g_{\Y}\}$} \\
         &\geq  \prof(\Y - g_{\Y}) 
\end{align*}
Hence, we obtain that $\efcount(X,Z)\le 1$, which is a contradiction. This establishes the theorem.
\end{proof}

\subsection{Fair Division {of Equal-Sized Goods}}
This section shows that the $\densestgreedy$ algorithm finds $\EFone$ allocations for
instances in which all the goods have equal sizes.\footnote{As in the general case, the values of the goods can be distinct. Here, we also retain the assumption that all the goods have distinct densities.}

\begin{theorem}
\label{thm:unit-size-ef1}
For any given budget-constrained fair division instance $\langle [m], [n], \{ v(g) \}_{g}, \{s(g)\}_{g}, \{B_a\}_{a} \rangle$ in which all the goods have the same size (i.e., $s(g) = s(g')$ for all $g, g' \in [m]$), Algorithm \ref{algo:dense} ($\densestgreedy$) computes an $\EFone$ allocation in polynomial time.
\end{theorem}
\begin{proof}
For instances in which the goods have the same size, ordering the goods in decreasing order of densities is equivalent to ordering them in decreasing order of values. Hence, the $\densestgreedy$ algorithm allocates the goods in decreasing order of value, while considering the budget constraints. Also, write $\beta$ to denote the common size of the goods and note that, in the current context, for any subset of goods $S \subseteq [m]$, we have $s(S)  = \beta \abs{S}$. 

We now estabilsh the $\EFone$ guarantee. Let $X$ be the set of goods allocated to an agent $a \in [n]$ and let $Y$ be the set of goods allocated to an agent $b \in [n]$, or to the charity at the end of Algorithm \ref{algo:dense}. We will prove that for any subset of goods $\Y \subseteq Y$, with $s(\Y) \leq B_a$, we have $\efcount(X,\Y) \leq 1$. Towards this, we consider the following two complementary and exhaustive cases \\
\noindent
Case {\sf 1}:  $\abs X \leq \abs{\Y}-1$. \\
\noindent
Case {\sf 2}: $\abs X \geq \abs{\Y}$. \\

\noindent
Case {\sf 1}: $\abs X \leq \abs{\Y}-1$. Write $Y=\{y_1,y_2,\dots,y_{\ell}\}$ to denote the goods in set $Y$, indexed in decreasing order of densities. Furthermore, let $y_{\ell'}$ denote the good in $\Y$ with minimum density. Equivalently, for the instance in hand, $y_{\ell'}$ is the minimum valued good in $\Y \subseteq Y$. Also, note that $\prof(\Y -y_{\ell'} ) \leq \prof \left(\Goods{Y}{\ell'-1} \right)$. 

Recall that $\beta$ denotes the size of each good. In the current case, we have $\abs X\le\abs{\Y}-1$ and, hence, $s(X)\le s(\Y)-\beta \le B_a-\beta=B_a-s\left(y_{\ell'}\right)$. That is, $s\left(X+y_{\ell'}\right)\le B_a$. Now, using Proposition \ref{greedyproperty-full}, we obtain $\prof(X)\ge\prof\left(\Goods{Y}{\ell'-1}\right)\ge\prof\left(Z-y_{\ell'}\right)$.
This shows that $\efcount(X,Z)\le 1$. \\

\noindent
Case {\sf 2}: $\abs X \geq \abs{\Y}$. Here, we will consider prefix subsets $\Goods{X}{i}$ and
$\Goods{\Y}{i}$ for all $i \in \{1,2,\ldots, |\Y|\}$ and establish, via induction over $i$,
that $\efcount\left(\Goods{X}{i},\Goods{\Y}{i}\right) \leq 1$.
This bound (with $i = |\Y|$) gives us the $\EFone$ guarantee.

For the base case, $i=1$, the stated bound $\efcount\left(\Goods{X}{i},\Goods{\Y}{i}\right) \leq 1$ holds directly,
since $\abs{\Goods{\Y}{1}}=1$. Now, for the induction step, assume that for an index $i \in \{2, 3, \ldots, \abs{Z}\}$, the envy count satisfies $\efcount\left(\Goods{X}{i-1},\Goods{\Y}{i-1}\right) \leq 1$.

Write $\structlb \coloneqq s(\Goods \Y {i-1})$ and $\structub \coloneqq s(\Goods \Y i)$. Since all the goods have the same size $\beta$,  we have $\structlb = (i-1) \beta$ and
$\structub = i \beta$. Furthermore, note that $\Goods{X}{i-1} = \Size{X}{\structlb}$ and $\Goods{X}{i} = \Size{X}{\structub}$. Therefore, the desired inequality $\efcount(\Goods{X}{i},\Goods{\Y}{i}) \leq 1$ is equivalent to $\efcount\left(\Size X \structub, \Size \Y \structub\right) \leq 1$.

By the induction hypothesis, we have
$\efcount\left(\Goods{X}{i-1},\Goods{\Y}{i-1}\right) \leq 1$, i.e.,
$\efcount\left(\Size X \structlb, \Size \Y \structlb\right) \leq 1$.
First, we note that if $\efcount\left(\Size X \structlb, \Size \Y \structlb\right) = 0$,
the induction step follows from Lemma \ref{lemma:lipschitz}:
\begin{align*}
    \efcount\left(\Size X \structub, \Size \Y \structub\right)
    \leq \efcount\left(\Size X \structlb, \Size \Y \structlb\right) + 1 = 1.
\end{align*}
In the complementary case, wherein $\efcount\left(\Size X \structlb, \Size \Y \structlb\right) = 1$,
we have $\prof\left(\Size X \structlb\right) < \prof\left( \Size \Y \structlb\right)$.
Write $x_i$ and $z_i$ to, respectively, denote the $i\Th$ good (indexed in decreasing order of densities)
in the sets $X$ and $\Y$. Since all the goods are of size $\beta$, we have
$s\left(\Size X \structlb\right) + s(z_i) = \structlb + s(z_i) = (i-1)  \beta + \beta = \structub \leq s(\Y) \leq B_a$.
That is, good $z_i$ can be included in $\Size X \structlb$ while maintaining agent $a$'s budget constraint.
Using this observation along with the inequality $\prof(\Size X \structlb) < \prof( \Size \Y \structlb)$
and Proposition \ref{greedyproperty}, we get $\dens(x_i) \geq \dens(z_i)$. This inequality reduces to $\prof(x_i) \geq \prof(z_i)$, since the goods have the same size. Furthermore, note that
$\Size X \structub = \Size X \structlb \cup \{x_i\}$ and $\Size \Y \structub = \Size \Y \structlb \cup \{z_i\}$. Hence, using $\prof(x_i)\ge\prof(z_i)$, we obtain
\begin{align*}
    \efcount\left(\Size X \structub, \Size \Y \structub\right)
    \leq \efcount\left(\Size X \structlb, \Size \Y \structlb\right) = 1.
\end{align*}
Therefore, for all $1\leq i \leq |\Y|$, we have  $\efcount(\Goods{X}{i},\Goods{\Y}{i}) \leq 1$ and,
in particular, $\efcount(X, \Y) \leq 1$. This gives use the stated $\EFone$ guarantee and completes the proof.
\end{proof}

\subsection{Fair Division in Cardinality Instances}
This section establishes the existence of $\EFone$ allocations in instances wherein each good has the same value.\footnote{As in the general case, the goods can have different sizes.} Note that, in such a setup, the densest good is the one with the smallest size. 
We establish the following theorem for cardinality instances. 

\begin{theorem} \label{theorem:cardinality}
For any given budget-constrained fair division instance $\langle [m], [n], \{ v(g) \}_{g}, \{s(g)\}_{g}, \{B_a\}_{a} \rangle$ in which all the goods have the same value (i.e., $v(g) = v(g')$ for all $g, g' \in [m]$), Algorithm \ref{algo:dense} ($\densestgreedy$) computes an $\EFone$ allocation in polynomial time.
\end{theorem}

Before proving Theorem \ref{theorem:cardinality}, we provide useful properties of the allocation returned by Algorithm \ref{algo:dense} for cardinality instances. First, note that the goods are allocated in the increasing order of their sizes, i.e., if goods $g_1$ and $g_2$ were, respectively, allocated in iterations $t_1$ and $t_2$ of the while-loop, with $t_1 < t_2$, then $s(g_1) \leq s(g_2)$. This inequality holds even among goods $g_1$ and $g_2$ that are allocated to different agents. Also, any good allocated to any agent has a smaller size than a good assigned to charity. \\
The following lemma provides another observation on the sizes of the allocated subsets.  
Let $X$ be the set of goods allocated to an agent $a \in [n]$ and let $Y$ be the set of goods allocated to an agent $b \in [n]$, or to the charity at the end of Algorithm \ref{algo:dense}.  

\begin{lemma} \label{round-robin_size}
For any index $i \leq \min\{ |X|, |Y|-1\}$, we have $s(X^{(i)}) \leq s(Y^{(i+1)})$.
\end{lemma}
\begin{proof}
For any index $j \leq \min \{ |X|, |Y|-1\}$, write $g_X^j$ to denote the  $j^{th}$ densest good in $X$ and  $g_Y^{j+1}$ to denote the $(j+1)\Th$ densest good in  $Y$. Also, let $\nu$ denote the common value of the goods. In this cardinality case, any agent's value after receiving $\ell$ goods is $\ell \nu$. 
If $g_Y^{j+1}$ is assigned to the charity, then $s(g_X^j) \leq s(g_Y^{j+1})$. If $g_Y^{j+1}$ is assigned to an agent $b \in [n]$, then we know that agent $a$ is assigned good $g_X^j$ before (i.e., in an earlier iteration) agent $b$ is assigned $g_Y^{j+1}$. As mentioned previously, the goods are assigned in increasing order of size.
Hence, for each $j \leq \min \{ |X|, |Y|-1\}$, we have  
\begin{align}
s(g_X^j) \leq s(g_Y^{j+1}) \label{ineq:goods-size}
\end{align}

Furthermore, consider index $i \leq \min\{ |X|, |Y|-1\}$. Summing inequality (\ref{ineq:goods-size}) over $j \in \{1, 2, \ldots, i\}$, we obtain $s(X^{(i)}) \leq s(Y ^{(i+1)})$. This proves the  lemma.
\end{proof}

We now establish Theorem \ref{theorem:cardinality}. 

\begin{proof}[Proof of Theorem \ref{theorem:cardinality}]
Let $X$ be the set of goods allocated to an agent $a \in [n]$ and let $Y$ be the set of goods allocated to an agent $b \in [n]$, or to the charity at the end of Algorithm \ref{algo:dense}. We will show that the $\EFone$ guarantee holds from set $X$ towards set $Y$. This establishes that the final allocation is $\EFone$, i.e., every agent $a \in [n]$ is $\EFone$ towards every other agent and the charity.

Write $\nu$ to denote the common value of the goods, and note that $\prof(X)  = \nu |X|$ and $\prof(Y) = \nu |Y|$. Hence, if $|X| \geq |Y|$, then the $\EFone$ guarantee holds. 

Therefore, for the rest of the proof we will consider the case $|X| < |Y|$. Let $F$ be any subset of $Y$, with size $s(F) \leq B_a$. We will prove that $|F| \leq |X| + 1$. Since all the goods have the same value $\nu$, this cardinality bound implies $\efcount(X, F) \leq 1$ and gives us the desired $\EFone$ guarantee. 

Assume, towards a contradiction, that $|F| \geq |X| + 2$ and let $f$ be the $(|X| + 2)\Th$-densest good in $F$. We will show that the existence of good $f$ contradicts Proposition \ref{greedyproperty-full} and, hence, complete the proof. In the current case, we have $|X| < |Y|$ and, hence, invoking Lemma \ref{round-robin_size} with index $i = |X|$, we obtain  
\begin{align}
s(X)  \leq s(Y^{(|X| + 1)}) \label{ineq:tm}
\end{align}
Note that increasing order of densities corresponds to decreasing order of sizes. Hence, the $(|X| + 1)$ most densest goods in $Y$---i.e., the goods that constitute $Y^{(|X|+1)}$---are in fact the ones with the smallest sizes. Hence, $s(Y^{(|X| + 1)}) \leq s(F^{(|X| + 1)})$. This bound and inequality (\ref{ineq:tm})  give us 
\begin{align*}
s(X) +  s(f) &\leq s(Y^{(|X| + 1)}) + s(f) \\
&\leq s(F^{(|X| + 1)}) + s(f) \\ 
&\leq B_a.
\end{align*} 
By definition, $f$ is the $(|X| + 2)\Th$ densest good in $F$. Hence, for some index $j \geq |X| +1$, good $f$ is the $(j+1)\Th$ densest good in $Y \supseteq F$. With this index $j$ in hand (i.e., with $h_{j+1} = f$), we invoke Proposition \ref{greedyproperty-full} to obtain $\prof(X) \geq \prof(Y^{(j)})$. Since the value of each good is $\nu$, the last inequality reduces to $\nu |X|  \geq j \nu$. This bound, however, contradicts the fact that $j \geq |X| +1$. Hence, it must be the case that $|F| \leq |X| +1$. As mentioned previously, this cardinality bound implies $\efcount(X, F) \leq 1$ and gives us the desired $\EFone$ guarantee. The theorem stands proved. 
\end{proof}

\subsection{Tightness of the Analysis}
\label{section:tightness}
In this section we provide an example for which Algorithm \ref{algo:dense} does not find an $\EFone$ allocation. This shows that  the $\EFtwo$ guarantee obtained for the algorithm (in Theorem \ref{thm:densestgreedy-ef2}) is tight.

We consider an instance with two agents and three indivisible goods, i.e., $n = 2$ and $m = 3$. Both the agents have a budget of one, $B_1 = B_2 = 1$. We set the sizes and values of the three goods as shown in the following table; here $\varepsilon \in (0,1/2)$ is an arbitrarily small parameter.  

\begin{center}
\begin{tabular}{|c|c|c|} 
 \hline
 Good & Size & Value \\ [0.3ex] 
\hline
$g_1$ & $\varepsilon$ & $10$ \\ 
 \hline
$g_2$ & $0.5$ & $0.5$ \\
 \hline
$g_3$ & $1- \varepsilon$ & $1 - 2\varepsilon$ \\ [1ex] 
\hline
\end{tabular}
\end{center}

The densities of the goods satisfy $\dens(g_1) > \dens(g_2) > \dens(g_3)$. Also, note that Algorithm \ref{algo:dense} returns the allocation with $A_1 = \{g_1, g_3\}$ and $A_2 = \{g_2\}$.  Since $v(g_1) > v(g_2)$ and $v(g_3) > v(g_2)$, the retuned allocation is not $\EFone$. 

\section{Conclusion and Future Work}
The current work makes notable progress towards efficient computation (and, hence, universal existence) of exact $\EFk$ allocations under budget constraints. Our algorithmic results are obtained via a patently simple algorithm, which lends itself to large-scale and explainable implementations. The algorithm's analysis, however, relies on novel insights, which are different from the ideas used for $\EFk$ guarantees in prior works and also from the ones used in approximation algorithms for the knapsack problem. 

In budget-constrained fair division (among agents with identical valuations) the existence and computation of $\EFone$ allocations is an intriguing open problem. We note that, interestingly, a constrained setting's computational (in)tractability does not reflect the fairness guarantee one can expect. For instance, the knapsack problem is {\rm NP}-hard for proportional instances and, at the same time, $\EFone$ allocations can be computed for such instances in polynomial time (Section \ref{section:efone-proportional-instances}). With this backdrop, obtaining $\EFk$ guarantees in the GAP formulation\footnote{As mentioned previously, in the GAP version of the problem, the goods have agent-specific sizes and values.} is another interesting direction for future work.

\bibliographystyle{alpha}
\bibliography{references}

\newcommand{\etalchar}[1]{$^{#1}$}
\begin{thebibliography}{ABFRV22}

\bibitem[ABFRV22]{amanatidis2022fair}
Georgios Amanatidis, Georgios Birmpas, Aris Filos-Ratsikas, and Alexandros~A
  Voudouris.
\newblock Fair division of indivisible goods: A survey.
\newblock {\em arXiv preprint arXiv:2202.07551}, 2022.

\bibitem[AKL21]{albers2021improved}
Susanne Albers, Arindam Khan, and Leon Ladewig.
\newblock Improved online algorithms for knapsack and gap in the random order
  model.
\newblock {\em Algorithmica}, 83(6):1750--1785, 2021.

\bibitem[ALMW22]{aziz2022algorithmic}
Haris Aziz, Bo~Li, Herve Moulin, and Xiaowei Wu.
\newblock Algorithmic fair allocation of indivisible items: A survey and new
  questions.
\newblock {\em arXiv preprint arXiv:2202.08713}, 2022.

\bibitem[BCE{\etalchar{+}}16]{handbook2016}
Felix Brandt, Vincent Conitzer, Ulle Endriss, J\'{e}r\^{o}me Lang, and Ariel~D.
  Procaccia.
\newblock {\em Handbook of Computational Social Choice}.
\newblock Cambridge University Press, 2016.

\bibitem[BCF{\etalchar{+}}22]{bilo2022almost}
Vittorio Bil{\`o}, Ioannis Caragiannis, Michele Flammini, Ayumi Igarashi,
  Gianpiero Monaco, Dominik Peters, Cosimo Vinci, and William~S Zwicker.
\newblock Almost envy-free allocations with connected bundles.
\newblock {\em Games and Economic Behavior}, 131:197--221, 2022.

\bibitem[BILS22]{bei2022price}
Xiaohui Bei, Ayumi Igarashi, Xinhang Lu, and Warut Suksompong.
\newblock The price of connectivity in fair division.
\newblock {\em SIAM Journal on Discrete Mathematics}, 36(2):1156--1186, 2022.

\bibitem[BKV18]{barman2018finding}
Siddharth Barman, Sanath~Kumar Krishnamurthy, and Rohit Vaish.
\newblock Finding fair and efficient allocations.
\newblock In {\em Proceedings of the 2018 ACM Conference on Economics and
  Computation}, pages 557--574, 2018.

\bibitem[BS21]{barman2021uniform}
Siddharth Barman and Ranjani~G Sundaram.
\newblock Uniform welfare guarantees under identical subadditive valuations.
\newblock In {\em Proceedings of the International Conference on International
  Joint Conferences on Artificial Intelligence}, pages 46--52, 2021.

\bibitem[Bud11]{Budish2011TheCA}
Eric Budish.
\newblock The combinatorial assignment problem: Approximate competitive
  equilibrium from equal incomes.
\newblock {\em Journal of Political Economy}, 119:1061 -- 1103, 2011.

\bibitem[CJS16]{cygan2016online}
Marek Cygan, {\L}ukasz Je{\.z}, and Ji{\v{r}}{\'\i} Sgall.
\newblock Online knapsack revisited.
\newblock {\em Theory of Computing Systems}, 58(1):153--190, 2016.

\bibitem[CK05]{chekuri2005polynomial}
Chandra Chekuri and Sanjeev Khanna.
\newblock A polynomial time approximation scheme for the multiple knapsack
  problem.
\newblock {\em SIAM Journal on Computing}, 35(3):713--728, 2005.

\bibitem[CKM{\etalchar{+}}19]{caragiannis2019unreasonable}
Ioannis Caragiannis, David Kurokawa, Herv{\'e} Moulin, Ariel~D Procaccia,
  Nisarg Shah, and Junxing Wang.
\newblock The unreasonable fairness of maximum nash welfare.
\newblock {\em ACM Transactions on Economics and Computation (TEAC)},
  7(3):1--32, 2019.

\bibitem[CKMS21]{chaudhury2021little}
Bhaskar~Ray Chaudhury, Telikepalli Kavitha, Kurt Mehlhorn, and Alkmini
  Sgouritsa.
\newblock A little charity guarantees almost envy-freeness.
\newblock {\em SIAM Journal on Computing}, 50(4):1336--1358, 2021.

\bibitem[DSR13]{deng2013story}
Yongheng Deng, Tien~Foo Sing, and Chaoqun Ren.
\newblock The story of singapore's public housing: From a nation of
  home-seekers to a nation of homeowners.
\newblock In {\em The future of public housing}, pages 103--121. Springer,
  2013.

\bibitem[FSTW19]{FluschnikSTW19}
Till Fluschnik, Piotr Skowron, Mervin Triphaus, and Kai Wilker.
\newblock Fair knapsack.
\newblock In {\em The Thirty-Third {AAAI} Conference on Artificial
  Intelligence, {AAAI}}, pages 1941--1948, 2019.

\bibitem[GGI{\etalchar{+}}21]{galvez2021approximating}
Waldo G{\'a}lvez, Fabrizio Grandoni, Salvatore Ingala, Sandy Heydrich, Arindam
  Khan, and Andreas Wiese.
\newblock Approximating geometric knapsack via l-packings.
\newblock {\em ACM Transactions on Algorithms (TALG)}, 17(4):1--67, 2021.

\bibitem[GGK{\etalchar{+}}21]{Galvez00RW21}
Waldo G{\'{a}}lvez, Fabrizio Grandoni, Arindam Khan, Diego
  Ram{\'{\i}}rez{-}Romero, and Andreas Wiese.
\newblock Improved approximation algorithms for 2-dimensional knapsack: Packing
  into multiple l-shapes, spirals, and more.
\newblock In {\em 37th International Symposium on Computational Geometry,
  SoCG}, volume 189, pages 39:1--39:17, 2021.

\bibitem[GLW21]{gan2021approximately}
Jiarui Gan, Bo~Li, and Xiaowei Wu.
\newblock Approximately envy-free budget-feasible allocation.
\newblock {\em arXiv preprint arXiv:2106.14446}, 2021.

\bibitem[GMT14]{gourves2014near}
Laurent Gourv{\`e}s, J{\'e}r{\^o}me Monnot, and Lydia Tlilane.
\newblock Near fairness in matroids.
\newblock In {\em ECAI}, pages 393--398, 2014.

\bibitem[GP15]{Spliddit}
Jonathan Goldman and Ariel~D. Procaccia.
\newblock Spliddit: Unleashing fair division algorithms.
\newblock {\em SIGecom Exch.}, 13(2):41--46, Jan 2015.

\bibitem[KMSW21]{K0001MSW21}
Arindam Khan, Arnab Maiti, Amatya Sharma, and Andreas Wiese.
\newblock On guillotine separable packings for the two-dimensional geometric
  knapsack problem.
\newblock In {\em 37th International Symposium on Computational Geometry,
  SoCG}, 2021.

\bibitem[KPP04]{knapbook}
Hans Kellerer, Ulrich Pferschy, and David Pisinger.
\newblock {\em Knapsack problems}.
\newblock Springer, 2004.

\bibitem[LMMS04]{LMMS042}
Richard~J Lipton, Evangelos Markakis, Elchanan Mossel, and Amin Saberi.
\newblock On approximately fair allocations of indivisible goods.
\newblock In {\em Proceedings of the 5th ACM Conference on Electronic
  Commerce}, pages 125--131, 2004.

\bibitem[Lue98]{lueker1998average}
George~S Lueker.
\newblock Average-case analysis of off-line and on-line knapsack problems.
\newblock {\em Journal of Algorithms}, 29(2):277--305, 1998.

\bibitem[MSV95]{marchetti1995stochastic}
Alberto Marchetti-Spaccamela and Carlo Vercellis.
\newblock Stochastic on-line knapsack problems.
\newblock {\em Mathematical Programming}, 68(1):73--104, 1995.

\bibitem[MT90]{martello1990knapsack}
Silvano Martello and Paolo Toth.
\newblock {\em Knapsack problems: algorithms and computer implementations}.
\newblock John Wiley \& Sons, 1990.

\bibitem[OSB10]{Othman2010}
Abraham Othman, Tuomas Sandholm, and Eric Budish.
\newblock Finding approximate competitive equilibria: Efficient and fair course
  allocation.
\newblock In {\em Proceedings of the 9th International Conference on Autonomous
  Agents and Multiagent Systems}, AAMAS '10, 2010.

\bibitem[Pis05]{pisinger2005hard}
David Pisinger.
\newblock Where are the hard knapsack problems?
\newblock {\em Computers \& Operations Research}, 32(9):2271--2284, 2005.

\bibitem[PKL21]{Patel0L21}
Deval Patel, Arindam Khan, and Anand Louis.
\newblock Group fairness for knapsack problems.
\newblock In {\em {AAMAS} '21: 20th International Conference on Autonomous
  Agents and Multiagent Systems}, pages 1001--1009, 2021.

\bibitem[PR20]{plaut2020almost}
Benjamin Plaut and Tim Roughgarden.
\newblock Almost envy-freeness with general valuations.
\newblock {\em SIAM Journal on Discrete Mathematics}, 34(2):1039--1068, 2020.

\bibitem[Ski99]{skiena1999interested}
Steven~S Skiena.
\newblock Who is interested in algorithms and why? lessons from the stony brook
  algorithms repository.
\newblock {\em ACM Sigact News}, 30(3):65--74, 1999.

\bibitem[ST93]{shmoys1993approximation}
David~B Shmoys and {\'E}va Tardos.
\newblock An approximation algorithm for the generalized assignment problem.
\newblock {\em Mathematical programming}, 62(1):461--474, 1993.

\bibitem[Suk21]{suksompong2021constraints}
Warut Suksompong.
\newblock Constraints in fair division.
\newblock {\em ACM SIGecom Exchanges}, 19(2):46--61, 2021.

\bibitem[WLG21]{wu2021budget}
Xiaowei Wu, Bo~Li, and Jiarui Gan.
\newblock Budget-feasible maximum nash social welfare is almost envy-free.
\newblock In {\em The 30th International Joint Conference on Artificial
  Intelligence (IJCAI 2021)}, pages 1--16, 2021.

\end{thebibliography}
\appendix
\section{The Distinct-Densities Assumption}
\label{appendix:distinct-densities}
As mentioned previously, for budget-constrained fair division, one can assume, without loss of generality, that the densities of the goods are distinct. We prove this assertion in the proposition below. 
\begin{proposition}
Given any fair division instance with budget constraints $\mathcal{I}=\langle [m], [n], \{ v(g) \}_{g\in [m]}, \{s(g)\}_{g \in [m]}, \allowbreak \{B_a\}_{a\in [n]} \rangle$, we can compute another instance $\mathcal{I}'=\langle [m], [n], \{ v'(g) \}_{g\in [m]}, \{s'(g)\}_{g \in [m]}, \{B_a'\}_{a\in [n]} \rangle$,
in polynomial time, with the properties that 
\begin{enumerate}
\item All the goods in $\mathcal{I}'$ have distinct densities.
\item If an allocation $\mathcal{A}$ is an $\EFtwo$ allocation in the constructed instance $\mathcal{I}'$, then $\mathcal{A}$ is an $\EFtwo$ allocation in $\mathcal{I}$ as well.
\end{enumerate}
\end{proposition}
\begin{proof}
From the instance $\mathcal{I}$, we obtain $\mathcal{I}'$ in two steps. First, we obtain an instance $\widehat{\mathcal{I}} =\langle [m], [n], \{ \widehat v(g) \}_{g\in [m]}, \allowbreak \{\widehat s(g)\}_{g \in [m]}, \left\{\widehat{B}_a \right\}_{a\in [n]} \rangle$ in which the values and sizes of all the goods are integral, $\widehat v(g)\in \mathbb{Z}_+$ and $\widehat{s}(g) \in\mathbb{Z}_+$, for all $g\in[m]$, and so are the agents' budgets $\widehat{B}_a \in \mathbb{Z}_+$, for all $a \in [n]$. Then, in the second step,
we obtain the desired instance $\mathcal{I}'$ from $\widehat{\mathcal{I}}$. Both the transformations---i.e., obtaining $\widehat{\mathcal{I}}$ from $\mathcal{I}$
and obtaining $\mathcal{I}'$ from $\widehat{\mathcal{I}}$---take polynomial time.

Recall that the values and sizes of all the goods in the given instance $\mathcal{I}$ are rational, $v(g) \in \mathbb{Q}_+$ and $s(g) \in \mathbb{Q}_+$ for all $g\in[m]$. To obtain instance $\widehat{\mathcal{I}}$ from $\mathcal{I}$ we simply scale the rational values ($v(g)$'s), sizes ($s(g)$'s), and budgets ($B_a$'s) such that they become integral. Note that, considering the rational representation of the inputs, we can efficiently find integers $\Gamma, \Gamma' \in \mathbb{Z}_+$, with polynomial bit-complexity, such that, for each good $g \in [m]$ and each agent $a \in [n]$ the scaled values and sizes are integers: $\widehat v(g) \coloneqq \Gamma \ v(g) \in \mathbb{Z}_+$ along with $\widehat s(g) \coloneqq  \Gamma' \ s(g) \in \mathbb{Z}_+$ and $\widehat{B}_a \coloneqq \Gamma' \ B_a \in \mathbb{Z}_+$. 
Such a scaling can be computed in polynomial time. 

Now, we perform the second transformation to obtain the desired instance $\mathcal{I}'$. Write $M\coloneqq m\prod_{g\in[m]}\widehat s(g)$. Furthermore, for every good $g\in[m]$ and and agent $a \in [n]$, let
\begin{align*}
v'(g)\coloneqq \widehat v(g)+\frac{1}{M^g}, \quad s'(g)\coloneqq \widehat s(g),\quad\text{ and }\quad B_a'\coloneqq \widehat{B}_a.
\end{align*}
This completes the construction of the instance $\mathcal{I}'$, and we will now prove the stated properties for $\mathcal{I}'$. 
First, note that each of $v'(g)$, $s'(g)$, $B_a'$ can be computed in polynomial time, since the factor $M \in \mathbb{Z}_+$ is of polynomial bit-complexity and the additive term $\frac{1}{M^g}$ can be computed in polynomial time as well.

Now, to prove that, in instance $\mathcal{I}'$, all the goods have distinct densities, consider any two goods $g, h \in [m]$ such that $g<h$. Since $M>\widehat s(g)$ and $M > \widehat{s}(h)$, we have
\begin{align*}
\frac{\widehat s(g)}{M^h} \neq \frac{\widehat s(h)}{M^g}, \  \quad \frac{\widehat s(g)}{M^h}\in(0,1), \quad \text{ and } \ \frac{\widehat s(h)}{M^g}\in(0,1).
\end{align*}
By construction, $\widehat v(\cdot)$ and $\widehat s(\cdot)$ are integer-valued functions. Hence, 
\begin{align*}
\widehat v(h)\widehat s(g)+\frac{\widehat s(g)}{M^h}\ne \widehat v(g)\widehat s(h)+\frac{\widehat s(h)}{M^g}.
\end{align*}
Therefore, we obtain
\begin{align*}
\frac{v'(h)}{s'(h)}=\frac{\widehat v(h)+(1/M^h)}{\widehat s(h)}\ne \frac{\widehat v(g)+(1/M^g)}{\widehat s(g)}=\frac{v'(g)}{s'(g)}.
\end{align*}
That is, all the goods in $\mathcal{I}'$ have distinct densities.

Next, we complete the proof by showing that, if $\mathcal{A}$ is an $\EFtwo$ allocation in $\mathcal{I}'$, then $\mathcal{A}$ is an $\EFtwo$ allocation in $\mathcal{I}$ as well. Let $X$ be a set of goods allocated to an agent $a\in[n]$ and $Y$ be a set of goods allocated to
another agent $b\in[n]$, or to the charity. 
Since $\mathcal{A}$ is an $\EFtwo$ allocation in $\mathcal{I}'$, for any subset $F \subseteq Y$, 
with $s'(F) \leq B'_a$, the following inequality holds, for two goods $f, f' \in F$:
\begin{align}
v'(X) \geq v'(F - f - f') \label{ineq:vprime}
\end{align}

In case $|F| \le 2$, then the $\EFtwo$ property is clearly satisfied. Hence, we assume that $|F|>2$. Furthermore, note that $s'(F) = \widehat{s}(F) = \Gamma' \ s(F)$ and $B'_a = \widehat{B}_a = \Gamma' \ B_a$. Therefore, subset $F \subseteq Y$ is considered for the $\EFtwo$ guarantee in $\mathcal{I}$ iff it is considered in $\mathcal{I}'$. This observation implies that, to establish the $\EFtwo$ guarantee for allocation $\mathcal{A}$ in instance $\mathcal{I}$, it suffices to show that $v(X) \geq v(F - f - f')$. We will prove this bound using inequality (\ref{ineq:vprime}).

Assume, towards a contradiction, that $v(X) < v(F-f - f')$. Equivalently, we have $\widehat{v}(X) < \widehat{v}(F-f - f')$; recall that $\widehat{v}(\cdot)$ is obtained by multiplicatively scaling $v(\cdot)$. Since $\widehat v(\cdot)$ is an integer-valued function, the last inequality reduces to 
\begin{align}
\widehat v(X) & \leq \widehat v(F- f - f') - 1 \nonumber \\
&<\widehat v(F- f - f' )- \left(\sum_{g\in X} \frac{1}{M^g} - \sum_{h\in (F-f-f') } \frac{1}{M^h}\right) \label{ineq:almost}
\end{align}
The last step follows from the fact that $|X| \leq m < M$. In addition, for any subset $Z \subseteq [m]$, we have $v'(Z) = \widehat{v}(Z) + \sum_{g \in Z} \frac{1}{M^g}$. Therefore, inequality (\ref{ineq:almost}) reduces to 
\begin{align*}
v'(X) < v'(F - f - f').
\end{align*}
This bound, however, contradicts inequality (\ref{ineq:vprime}). Therefore, it must be the case that $v(X) \geq v(F - f - f')$, i.e., the $\EFtwo$ guarantee holds for allocation $\mathcal{A}$ in the underlying instance $\mathcal{I}$ as well. This completes the proof.  
\end{proof}

\section{Missing Proofs from Section \hyperref[section3:densestgreedy-ef2]{3} }
\label{appendix:proposition-proofs}
\subsection{Proof of Proposition \hyperref[greedyproperty]{1} }

\PropGreedyOne*

\begin{proof}
We address the cases when $Y$ is a set of goods allocated to one of the agents $b \in [n]$ or $Y$ is the goods assigned to charity, separately. \\
First, consider the case wherein $Y = A_b$ for some agent $b \in [n]$.
We note that in each iteration of the while-loop in Algorithm \ref{algo:dense}, either a good is assigned to a selected agent or an agent is marked as inactive. We will write $t$ to denote the iteration count of the while-loop. Furthermore, define $a_t$ to be the agent chosen in the $t^{th}$ iteration (the minimum-valued active agent at that point) and $A^t_a$ to denote the subset of goods assigned to agent $a \in [n]$ till the $t$th iteration.

A useful observation about Algorithm \ref{algo:dense} is that the assignment of a good to an agent is permanent, i.e., a good once assigned to an agent remains with the agent even in the final allocation. This also implies that the value of the goods allocation to an agent never decreases as the algorithm progresses, i.e., $\prof\left(A^{t}_a \right) \leq \prof\left(A^{t+1}_a \right)$, where the inequality holds if and only if a good was assigned to agent $a$ in the $(t+1)^{th}$ iteration.

We first show that $\prof\left(X^{(i)} \right) < \prof\left(Y^{(j)}\right)$ implies that the good $g_{i+1}$ is assigned to agent $a$ in an earlier iteration than when  the good $h_{j+1}$ is assigned to agent $b$. Let the iteration count when  the good $g_{i+1}$ is assigned to agent $a$ and  the good $h_{j+1}$ is assigned to agent $b$ be $t_1$ and $t_2$, respectively. Assume, towards a contradiction, that $t_1 > t_2$. Using the above observation, we get $\prof\left(X^{t_2} \right) \leq \prof\left(X^{t_1} \right) < \prof\left(Y^{(j)} \right) = \prof\left(Y^{t_2} \right)$. However, as the algorithm in Line \ref{line3} selects the minimum valued active agent, we have  $\prof\left(Y^{t_2} \right) \leq \prof\left(X^{t_2} \right)$, which is a contradiction. Hence, it must be the case that $t_1 < t_2$. This inequality, along with the definition of $t_1$ and $t_2$, imply that $h_{j+1}$ was unassigned at the beginning of iteration count $t_1$. 

Using $s\left(X^{(i)} + h_{j+1}\right) \leq B_a$ and the fact that the algorithm in Line \ref{line:assign-good} selects the densest good that fits, we conclude that $\dens(g_{i+1}) > \dens(h_{j+1})$. Otherwise, the good $h_{j+1}$ would be the densest good that fits in the budget of agent $a$.

For the second case, let $Y$ be the goods assigned to the charity, i.e., $Y = [m] \backslash \cup_{i=1}^n A_i$. We prove a stronger claim for this case. Let $h$ be an arbitrary good in $Y$ and suppose that for some $i<|X|$, we have $s\left(X^{(i)} + h\right) \leq B_a$. Then we have $\dens(g_{i+1}) > \dens(h)$. Let the iteration count when  the good $g_{i+1}$ is assigned to agent $a$ be $t_1$. We know that $h$ was assigned before this iteration. Recall that all the leftover goods are assigned to charity at the end of $\densestgreedy$. Using $s\left(X^{(i)} + h\right) \leq B_a$ and the fact that the algorithm in Line \ref{line:assign-good} selects the densest good that fits, we conclude that $\dens(g_{i+1}) > \dens(h)$. Otherwise, the good $h$ would be the densest good that fits in the budget of agent $a$.\\
This completes the proof.
\end{proof}

\subsection{Proof of Proposition \hyperref[greedyproperty-full]{2} }

\PropGreedyTwo*

\begin{proof}
Like in the proof of Proposition \hyperref[greedyproperty]{1}, we deal with the cases when $Y$ is the set of goods assigned to an agent $b \in [n]$ or $Y$ is the set of goods assigned to charity, separately. \\

First, consider the case $Y = A_b$, for some agent $b \in [n]$.
Write $t$ to denote the iteration count when the good $h_{j+1}$ was assigned to agent $b$. We first prove that agent $a$ was active in the beginning of the $t^{th}$ iteration. Assume, towards a contradiction, that agent $a$ was inactive at the beginning of the $t^{th}$ iteration. Let $t' < t$ be the iteration when agent $a$ was marked as inactive. From the assumption, we have  $s\left(A_a+h_{j+1}\right)\le B_a$, i.e., the good $h_{j+1}$ fits in agent $a$'s budget. This, however, contradicts the fact that, in Algorithm \ref{algo:dense}, an agent is marked inactive only when it is the minimum-valued agent and no unassigned good fits in its budget -- here, $a$ is marked inactive even though $h_{j+1}$ fits into her budget. Hence, agent $a$ must have been active at the beginning of the $t^{th}$ iteration. 

As the algorithm in Line \ref{line3} selects the minimum valued active agent, we have  $\prof\left(X\right)\geq \prof\left(Y^{(j)}\right)$; recall that the valuation of agent $a$ does not decrease as the algorithm progresses. 

For the second case, let $Y$ be the goods assigned to the charity, i.e., $Y = [m] \backslash \cup_{i=1}^n A_i$. Write $t$ to denote the iteration count when agent $a$ is marked inactive. The good $h_{j+1}$ was unasigned before iteration $t$, as it remains in charity even after all the agents are marked inactive. This, however, contradicts the fact that in Algorithm  \ref{algo:dense}, an agent is marked inactive only when it is the minimum-valued agent and no unassigned good fits in its budget -- here, $a$ is marked inactive even though $h_{j+1}$. Hence, we get a contradiction to the existence of a good $h_{j+1}$ with the aformentioned properties.

This completes the proof of Proposition \hyperref[greedyproperty-full]{2}.

\end{proof}

\section{$\EFtwo$ for Goods with Agent-Specific Sizes}
\label{appendix:mini-GAP}
This section addresses budget-constrained fair division settings wherein the goods' sizes are agent-specific. We obtain the $\EFtwo$ guarantee for such settings via a direct   generalization of Algorithm \ref{algo:dense} ($\densestgreedy$). 

In this section, an instance of the budget-constrained fair division problem is a tuple 
$\langle [m], [n], \{ v(g) \}_{g}, \allowbreak \{s_a(g)\}_{a,g}, \{B_a\}_{a} \rangle$. 
Here, we need to assign $m$ indivisible goods among $n$ agents and the charity. 
The agents' valuations are identical and additive; in particular, every agent values each good $g \in [m]$ at $v(g)$. 
Sizes of the goods, however, are nonidentical: each good $g \in [m]$ has a size $s_a(g) \in \mathbb{Q}_+$ with respect to the agent $a$. 
Each agent $a\in[n]$ should be allocated a subset of goods $A_a\subseteq[m]$ such that
$A_1,A_2,\dots,A_n$ are mutually disjoint. Furthermore, for an agent $a$,
the assigned bundle $A_a \subseteq [m]$ must be of total size at most the agent's budget $B_a \in \mathbb{Q}_+$, 
i.e., the assigned bundle $A_a$ satisfies $s_a(A_a) = \sum_{g \in A_a} s_a(g) \leq B_a$. 
The set of remaining goods, $[m]\setminus \cup_{i=1}^n A_i$, is assigned to the charity.

We define the notion of $\EFtwo$ for this setting. The solution concept here generalizes Definition \ref{defn:eftwo}.  
\begin{definition}[\EFtwo]
Let $\mathcal{A} = (A_1, A_2, \ldots, A_n)$ be an arbitrary allocation.
An agent $a\in[n]$ is said to be \emph{envy-free up to two goods} ($\EFtwo$) towards agent $b\in[n]$
iff for every subset $F \subseteq A_b$, with $s_a(F) \leq B_a$ (and $|F| \geq 2$), there exist
goods $f_1,f_2 \in F$ such that $v(A_a) \geq v(F \setminus \{f_1,f_2\})$. Further, an agent $a\in[n]$ is said to be $\EFtwo$ towards the charity
iff for every subset $F \subseteq [m]\setminus \cup_{a=1}^n A_a$, with $s_a(F) \leq B_a$ (and $|F| \geq2$), there exist goods 
$f_1,f_2 \in F$ such that $v(A_a) \geq v(F \setminus \{f_1,f_2\})$.
The allocation $\mathcal A$ is said to be $\EFtwo$ iff every agent $a\in[n]$ is $\EFtwo$ towards every other agent $b\in[n]$ and the charity.
\end{definition}

Agents perceive the size of a good differently and, hence, the density of a good is also agent-specific. To accommodate this, we denote the density of a good $g \in [m]$, with respect to an agent $a \in [n]$, as $\dens_a(g) \coloneqq \frac{v(g)}{s_a(g)}$. Since the densities are agent-specific, the notion of the densest good across agents is not well formed.  However, the notion of the minimum-valued agent \emph{does} exist, since the agents' valuations are identical. Furthermore, the function $\efcount(\cdot)$ is well-defined (see equation (\ref{eqn:efcount-defn})).

We obtain the $\EFtwo$ guarantee in the current context via Algorithm \ref{algo:dense_gap}. As mentioned previously, this algorithm is obtained by generalizing the $\densestgreedy$ algorithm. The key difference between the two algorithms is in Line \ref{line:assign-good_gap} of Algorithm \ref{algo:dense_gap}, where we allocate the densest good (that fits) from the selected agent $a$'s perspective, i.e., densest with respect to $\dens_a(\cdot)$.

The analysis of Algorithm \ref{algo:dense_gap} is similar to that of the $\densestgreedy$ algorithm. Though, multiple arguments (e.g., Lemma \ref{lemma_gap}) are more involved. For a self-contained treatment, we provide a complete analysis and even repeat the common parts. Formally, we establish the following guarantee 
\begin{theorem}
\label{thm:densestgreedy-ef2_gap}
For any given fair division instance with budget constraints $\langle [m], [n], \{ v(g) \}_{g\in [m]}, \{s_a(g)\}_{a \in [n],g \in [m]}, \allowbreak \{B_a\}_{a\in [n]} \rangle$, Algorithm \ref{algo:dense_gap} computes an $\EFtwo$ allocation in polynomial time.
\end{theorem}

\begin{algorithm}
\caption{Given instance $\langle [m], [n], \{ v(g) \}_{g}, \{s_a(g)\}_{a,g}, \allowbreak \{B_a\}_{a} \rangle$,  
allocate the goods $[m]$ among agents $[n]$ (the unassigned goods go to charity).}  \label{algo:dense_gap}
\begin{algorithmic}[1]
\STATE Initialize allocation $(A_1, \dots, A_{n}) \leftarrow (\emptyset, \ldots, \emptyset)$. Also, define set of active agents $N \coloneqq [n]$ and set of unallocated goods $G \coloneqq [m]$. 
\WHILE{$G \neq \emptyset$}
\STATE Select arbitrarily a minimum-valued agent $a \in N$, i.e., $a = {\displaystyle \arg \min_{b \in N} \prof(A_b)}$.
\IF{for all goods $g\in G$ we have $s_a\left( A_a + g \right) > B_a$}
\STATE Set agent $a$ to be inactive, i.e., $N \leftarrow N \setminus \{a\}$.
\ELSE
\STATE \label{line:assign-good_gap} Write $g' = {\displaystyle \argmax_{g \in G :\  s_a(A_a + g) \leq B_a} \dens_a(g)}$ 
and update $A_a \leftarrow A_a + g'$ along with $G \leftarrow G - g'$.
\ENDIF
\ENDWHILE
\RETURN $(A_1,A_2,\dots,A_{n})$
\end{algorithmic}
\end{algorithm}

In every iteration of the while-loop in Algorithm \ref{algo:dense_gap}, either a good is allocated to an agent or an agent is marked as inactive. Let $m'$ denote the number of goods allocated to the agents (the unassigned goods are assigned to charity at the end by default). Throughout this section, we will write $\sigma \in \mathbb{S}_{m'}$ to denote the order in which the goods are allocated (among the agents) by the algorithm. That is, $\sigma: [m'] \rightarrow [m']$ is the  permutation with the property that, for any two indices $t, t' \in [m']$, with $t < t'$, the good $\sigma(t)$ was allocated in an earlier iteration (of Algorithm \ref{algo:dense_gap}) than good $\sigma(t')$.\footnote{Note that goods $\sigma(t)$ and $\sigma(t')$ might have been assigned to different agents.}  

We now define prefix-subsets for the current context. For any subset of goods $S = \{s_1, s_2, \ldots, s_k \}$, indexed according to the allocation order $\sigma$, and any index $1 \leq i \leq |S|$, write $S^{(i)} \coloneqq \{s_1, \ldots, s_i \}$.
\begin{definition}[Prefix Subset $\Size{S_a}{B}$]
\label{definition:prefix-subset_gap}
For any agent $a \in [n]$ and a subset of goods $S = \{g_1, g_2, \ldots, g_k \}$, indexed according to the allocation order $\sigma$, and for any threshold $B < s_a(S)$, let $P=\{g_1, \ldots, g_{\ell-1} \}$ be the (cardinality-wise) largest prefix of $S$ such that $s_a(P) \leq B$. Then, we define $S_a^{[B]} \coloneqq P \cup \left\{ \alpha \cdot g_{\ell} \right\}$, where $\alpha = \frac{B-s_a(P)}{s_a(g_{\ell})}$.  
\end{definition}

If the threshold $B \geq s_a(S)$, then we set $\Size{S_a}{B}=S$. The following proposition directly follows from the selection criteria of Algorithm \ref{algo:dense_gap}.
\begin{proposition}\label{greedyproperty_gap}
Let $A_a=\{g_1, g_2, \ldots, g_k\}$ and $A_b=\{h_1, h_2, \ldots, h_\ell\}$ denote, respectively, the sets of goods assigned to agents $a, b \in [n]$ at the end of Algorithm \ref{algo:dense_gap}; the goods in these sets are indexed according to $\sigma$ (i.e., in order of allocation in Algorithm \ref{algo:dense_gap}).  Also, for indices $i < |A_a|$ and $j < |A_b|$, suppose  $\prof\left(A^{(i)}_a \right) < \prof\left(A^{(j)}_b\right)$ and $s_a\left(A^{(i)}_a + h_{j+1}\right) \leq B_a$. Then, $\dens_a(g_{i+1}) > \dens_a(h_{j+1})$. 
\end{proposition}


The following three lemmas provide generalizations of the ones provided in the Section \ref{section:structural-props}. 

\begin{lemma}\label{lemma:lipschitz_gap}
For any agent $a \in [n]$, any subset of goods $X$ and $Y$ along with any index $i < |Y|$, let $\structlb \coloneqq s_a\left(Y^{(i)}\right)$ and $\structub \coloneqq s_a\left(Y^{(i+1)}\right)$. 
Then, $\efcount\left(\Size{X_a}{\structub}, \Size{Y_a}{\structub}\right) \leq \efcount\left(\Size{X_a}{\structlb}, \Size{Y_a}{\structlb}\right) + 1$.
\end{lemma}
\begin{proof}
Write $c \coloneqq \efcount\left(\Size{X_a}{\structlb}, \Size{Y_a}{\structlb}\right)$. Therefore, by definition, there exists a size-$c$ subset $R \subseteq \Size{Y_a}{\structlb}$ with the property that 
$\prof(\Size{X_a}{\structlb}) \geq \prof(\Size{Y_a}{\structlb} \setminus R)$. Define subset $R' \coloneqq  R \cup \{h_{i+1}\}$,  where $h_{i+1}$ is the good in the set $Y^{(i+1)}\setminus Y^{(i)}$. For this set $R'$ of cardinality $c+1$, we have
\begin{align*}
\prof\left( \Size{X_a}{\structub} \right) &\geq \prof\left( \Size{X_a}{\structlb} \right) 
     \geq \prof\left(\Size{Y_a}{\structlb} \setminus R\right) 
     = \prof\left(\Size{Y_a}{\structub} \setminus R'\right).
\end{align*}
This implies $\efcount\left(\Size{X_a}{\structub}, \Size{Y_a}{\structub}\right)\leq c + 1$, and the lemma stands proved. 
\end{proof}

\begin{lemma}\label{existenceof2cut_gap}
Let $a \in [n]$ be an agent and $X$ and $Y$ be any subsets of goods with the property that $\efcount(X,Y) \geq 2$. Then, there exists an index $t \leq |Y|$ such that, with $\structlb \coloneqq s_a\left(\Goods{Y}{t}\right)$, we have  $\efcount\left(\Size{X_a}{\structlb}, \Size{Y_a}{\structlb}\right) = 2$.
\end{lemma}
\begin{proof}
The lemma follows from a discrete version of the intermediate value theorem. For indices $t \in \{0,1,2,\ldots, |Y|\}$, define the function $h(t) \coloneqq s_a\left(\Goods{Y} {t}\right)$. Extending this function, we consider the envy count at different size thresholds; in particular, write $H(t) \coloneqq \efcount\left(\Size{X_a}{h(t)}, \Size{Y_a}{h(t)}\right)$ for each $t \in \{0,1,2,\ldots, |Y|\}$. Note that $H(0) = 0$. We will next show that $(i)$ $H(|Y|) \geq 2$ and $(ii)$ the discrete derivative of $H$ is at most one, i.e., $H(t+1) - H(t) \leq 1$ for all $0 \leq t < |Y|$. These properties of the integer-valued function $H$ imply that there necessarily exists an index $t^*$ such that $H(t^*) = 2$. This index $t^*$  satisfies the lemma.  

Therefore, we complete the proof by establishing properties $(i)$ and $(ii)$ for the function $H(\cdot)$. For $(i)$, note that the definition of the prefix subset gives us $\prof(X)\ge \prof\left(\Size{X_a}{s_a(Y)}\right)$. Hence, $\efcount\left(\Size {X_a} {s_a(Y)}, Y\right) \geq \efcount(X, Y) \geq 2$; the last inequality follows from the lemma assumption.  Since $h(|Y|) = s_a(Y)$, we have $H(|Y|) \geq 2$. Property $(ii)$ follows directly from Lemma \ref{lemma:lipschitz_gap}. This completes the proof. 
\end{proof}

The following lemma essentially asserts that if we have two subsets $X'$ and $\Y'$ with $\efcount(X',\Y') = 2$ and one adds more value into $X'$ than $\Y'$, then the envy count does not increase. 

\begin{lemma}\label{envy_diff_gap}
Given an agent $a\in [n]$ and two subsets of goods $X$ and $\Y$ along with two nonnegative size thresholds $\structlb,\structub \in \mathbb{R}_+$ with the properties that  
\begin{itemize}
\item $\efcount\left(\Size{X_a}{\structlb}, \Size {\Y_a}{\structub}\right) = 2$  and 
\item $\prof\left(X \setminus \Size {X_a} {\structlb}\right) \geq  \prof\left(\Y \setminus \Size {\Y_a} {\structub}\right)$.
\end{itemize}
Then, $\efcount(X,\Y) \leq \efcount\left(\Size {X_a} {\structlb}, \Size {\Y_a} {\structub}\right)=2$.
\end{lemma}
\begin{proof}
Given that $\efcount\left(\Size {X_a} {\structlb}, \Size {\Y_a} {\structub}\right) = 2$, there exist two goods $g_1', g_2' \in \Size {\Y_a} {\structub}$
such that $\prof\left(\Size{\Y_a}{\structub} - g_1' - g_2'\right)  \leq \prof\left(\Size {X_a} {\structlb}\right)$. Now, using the definition of the prefix subsets (Definition \ref{definition:prefix-subset_gap}) we get
\begin{align}
	\prof(X) &= \prof\left(\Size {X_a} {\structlb}\right) + \prof\left(X \setminus \Size {X_a} {\structlb}\right) \nonumber \\
	& \geq \prof\left(\Size {\Y_a} {\structub} - g_1' - g_2'\right) + \prof\left(X \setminus \Size {X_a} {\structlb}\right) \nonumber \\
	& \geq \prof\left(\Size {\Y_a} {\structub} - g_1' - g_2'\right) +  \prof\left(\Y \setminus \Size {\Y_a} {\structub}\right) \tag{via lemma assumption} \\
	& = \prof(\Y)  -\left( \prof(g'_1) + \prof(g'_2) \right) \label{ineq:interim_gap}
\end{align}
The definition of the prefix subset $\Size{\Y_a}{\structub}$ ensures that, corresponding to goods $g'_1, g'_2 \in \Size{\Y_a}{\structub}$, there exist two goods $g_1, g_2 \in \Y$ such that $v(g_1) + v(g_2) \geq v(g'_1) + v(g'_2)$. This bound and inequality (\ref{ineq:interim_gap}) give us $\prof(X) \geq \prof(\Y - g_1 - g_2)$. This implies $\efcount(X,\Y) \leq 2$ and completes the proof of the lemma. 
\end{proof}

\subsection{Proof of Theorem \ref{thm:densestgreedy-ef2_gap}}
This section establishes Theorem \ref{thm:densestgreedy-ef2_gap}, i.e., the allocation returned by \cref{algo:dense_gap}
is $\EFtwo$. To show this, we first prove that every agent is $\EFone$ (a stronger notion than $\EFtwo$) towards charity.
Then we show the $\EFtwo$ property between the agents, thus completing the proof.

To show the $\EFone$ guarantee against the charity, we use the following proposition.
\begin{proposition}\label{greedyproperty_gap-charity}
Let $A_a=\{g_1, g_2, \ldots, g_k\}$ be the set of goods assigned to an agent $a\in[n]$,
indexed according to $\sigma$,
and let $h$ be a good left unassigned, i.e., given to the charity, at the end of Algorithm \ref{algo:dense_gap}.
Then, for any index $i<k$, if $\rho_a(h)>\rho_a(g_{i+1})$, then $s_a\left(A_a^{(i)}+h\right)> B_a$.
\end{proposition}

\begin{lemma}
Let $A_a$ be the set of goods assigned to an agent $a\in[n]$ and let $S$ be some subset of goods assigned to \emph{charity} by \cref{algo:dense_gap} such that $s_a(S) \leq B_a$. Then $\efcount(A_a,S) \leq 1$, i.e., the $\EFone$ guarantee holds for any agent towards charity.
\label{envy-against-charity}
\end{lemma}

\begin{proof}
Let $\widehat{g} \in S$ be the densest good in $S$ according to agent $a$, i.e., $\widehat{g} = \argmax_{g \in S} \dens_a(g)$. We show that $v(A_a) \geq v(S - \widehat{g})$.
It follows from the definition of $\widehat{g}$ that $v(S) \leq \dens_a(\widehat{g}) \cdot B_a$. To prove the lemma, we show that $v(A_a) + v(\widehat{g}) \geq \dens_a(\widehat{g}) \cdot B_a \geq v(S)$. Towards this, define $X$ to be the subset of goods in $A_a$ that are denser than $\widehat{g}$. Recall that in \cref{algo:dense_gap}, each agent is assigned goods in decreasing order of density according to her. By \cref{greedyproperty_gap-charity}, we know that
$s_a(X+\widehat{g}) = s_a\left(A_a^{(|X|)}+\widehat{g}\right)> B_a$. \\
Using the fact that we have $\dens_a(g) \geq \dens_a(\widehat{g})$ for all the goods $g \in X$, we write
\begin{align*}
v(A_a) + v(\widehat{g}) &\geq v(X) + v(\widehat{g}) \\
			&\geq \dens_a(\widehat{g}) \cdot (s_a\left(X\right) + s_a(\widehat{g})) \tag{ using $\dens_a(X) > \dens_a(\widehat{g})$ }\\
			& \geq  \dens_a(\widehat{g}) \cdot B_a \tag{using $s_a(X) +s_a(\widehat{g}) > B_a$}\\
			& \geq  \dens_a(\widehat{g}) \cdot s_a(A_a)  \\
			& \geq v(S) \tag{using $\dens_a(\widehat{g}) \geq \dens_a(S)$} \\
\end{align*}

We have, $v(A_a) \geq v(S) - v(g')$ and this completes the proof for the $\EFone$ guarantee for any agent towards the charity.
\end{proof}

Now we prove that the $\EFtwo$ guarantee holds between any two agents. Fix any two agents $a, b \in [n]$, and let $A_a$ and $A_b$ be the subsets of goods allocated to them, respectively, at the end of Algorithm \ref{algo:dense_gap}.
For ease of notation, we will denote $A_a$ by $X = \{x_1, x_2, \dots, x_k \}$ and $A_b$ by $Y =  \{y_1, y_2, \dots, y_{\ell} \}$; the goods in both these sets are indexed according to the allocation order $\sigma$. Proving $\EFtwo$ between the two agents corresponds to showing that, for any subset of goods $\Y \subseteq Y$, with $s_a(\Y) \leq B_a$, we have $\efcount(X,\Y) \leq 2$. 

Consider any such subset $\Y$ and index its goods in order of $\sigma$. Note that, if $\efcount(X, \Y) \leq 1$, we already have the $\EFtwo$ guarantee. Therefore, in the remainder of the proof we address the case wherein $\efcount(X, \Y) \geq 2$. We will in fact show that this inequality cannot be strict, i.e., it must hold that the envy count is at most $2$ and, hence, we will obtain the $\EFtwo$ guarantee.

We start by considering function $h_a(i)$ which denotes the size---with respect to $s_a(\cdot)$---of the first (order according to the $\sigma$) $i$ goods in set $Z$, i.e., $h_a(i) \coloneqq s_a(\Goods{Z}{i})$ for $i \in \{0,1,2,\ldots, |Z|\}$. Furthermore, define index 
\begin{align}
\label{eq:t-repeat}
t \coloneqq \min\left\{ i : \efcount\left(\Size{X_a}{h_a(i)}, \Goods{\Y}{i}\right) = 2\right\}
\end{align}

\begin{figure}[h]
\centering
\includegraphics[scale=1]{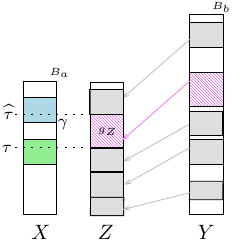} 
\caption{Figure illustrating size thresholds $\tau$, $\tauprime$, and the good $g_{\Y}$ with respect to agent $a$.} \label{figure:cut-offs_gap}
\end{figure}
Existence of such an index $t \ge 2$ follows from Lemma \ref{existenceof2cut_gap}. Also, note that $\Goods{\Y}{i} = \Size{\Y_a}{h_a(i)}$. 
We will denote the $t^{th}$ good in $\Y$ by $g_{\Y}$, i.e., $g_{\Y} = z_{t}$. In addition, using $t$ we define the following two size thresholds (see Figure \ref{figure:cut-offs_gap})
\begin{align}
\tau \coloneqq s_a\left(\Goods \Y {t-1}\right) \quad \text{ and } \quad \tauprime \coloneqq s_a\left(\Goods \Y t\right)\label{eq:defn-tau_gap}
\end{align}
That is, $\tau = h_a(t-1)$ and $\tauprime = h_a(t)$. Now, Lemma \ref{lemma:lipschitz_gap} and the definition of $t$ (equation (\ref{eq:t-repeat})) give us $\efcount\left(\Size {X_a} {\tau}, \Size {\Y_a} {\tau}\right) \ge 1$. Furthermore, using the minimality of $t$ we get $\efcount\left(\Size {X_a} {\tau}, \Size {\Y_a} {\tau}\right) < 2$.
Hence, 
\begin{align}
\efcount\left(\Size {X_a} {\tau}, \Size {\Y_a} {\tau}\right) = 1 \label{eqn:xtau-ztau_gap}
\end{align}

We will establish two properties for the sets $X$ and $Z$ under consideration and use them to invoke \cref{envy_diff_gap}. Specifically, in \cref{small_ef2_gap} we will show that  $\efcount\left(\Size {X_a} {\tau}, \Size {\Y_a} \tauprime\right) = 2$ and in \cref{lemma_gap} we prove $\prof\left(X \setminus \Size {X_a} \tau\right) \geq \prof\left(\Y \setminus \Size {\Y_a} \tauprime\right)$. These are exactly the two properties required to apply \cref{envy_diff_gap} with $  \structlb = \tau$ and $ \structub = \tauprime$. 

\begin{lemma}\label{small_ef2_gap}
$\efcount\left(\Size {X_a} {\tau}, \Size {\Y_a} \tauprime\right) = 2$.
\end{lemma}
\begin{proof}
Since $\efcount\left(\Size {X_a} \tau, \Size {\Y_a} \tau\right) = 1$ (see equation (\ref{eqn:xtau-ztau_gap})), there exists a good $g_1 \in \Size {\Y_a} \tau$ such that $\prof\left(\Size{X_a}{\tau}\right) \geq \prof\left(\Size{\Y_a}{\tau} - g_1\right)$.  
Also, by definition, we have have $ \Size{\Y}{\tauprime} = \Size{\Y}{\tau} \cup \{ g_{\Y} \}$. Hence, the previous inequality reduces to $\prof\left(\Size{X}{\tau}\right) \geq \prof\left(\Size \Y \tauprime -  g_{\Y} - g_1\right)$. That is, removing $g_1$ and $g_{\Y}$ from $\Size{\Y}{\tauprime}$ gives us a set with value at most that of $\Size{X}{\tau}$. Therefore, we have $\efcount\left(\Size {X_a} {\tau}, \Size {\Y_a} \tauprime\right) = 2$. The lemma stands proved. 
\end{proof}

We define $\gamma$ as the size of the goods in $X$ that are at least as dense as $g_{\Y}$ with respect to agent $a$, i.e.,  
\begin{align}
\gamma \coloneqq \sum_{g \in X: \dens_a(g) \geq \dens_a(g_{\Y})} s(g) \label{eq:def-g_gap}
\end{align}

We will establish bounds considering $\gamma$ and use them to prove Lemma \ref{lemma_gap} below. 

\begin{claim} \label{lemma5_gap}
It holds that $\gamma \leq \tauprime$ and $\prof\left(\Size {X_a} \gamma\right) < \prof\left(\Size {\Y_a} \tau\right)$. 
\end{claim}
\begin{proof}
We will first establish the stated upper bound on $\gamma$. Assume, towards a contradiction, that $\gamma > \tauprime$. By definition of $\gamma$, we have that all the goods in $\Size{X_a}{\gamma}$ have density at least $\rho(g_{\Y})$. Now, given that $\gamma > \tauprime$, we get that the density of each good in $\Size{X}{\tauprime}$ is at least $\rho(g_{\Y})$. In particular, all the goods in the set $\Size{X_a}{\tauprime} \setminus \Size{X_a}{\tau}$ are at least as dense as $g_{\Y}$ with respect to agent $a$. Hence, $\prof\left(\Size {X_a} \tauprime \setminus \Size {X_a} \tau\right) \geq \prof\left(\Size {\Y_a} \tauprime \setminus \Size {\Y_a} \tau\right) = \prof(g_{\Y})$. This inequality and equation (\ref{eqn:xtau-ztau_gap}) give us $\efcount(\Size{X_a}{\tauprime}, \Size{\Y_a}{\tauprime}) \leq 1$; see Lemma \ref{envy_diff_gap}. This bound, however, contradicts the definition of $t$ (and, correspondingly, $\tauprime$) as specified in equation (\ref{eq:t-repeat}). This gives us the desired upper bound, $\gamma \leq \tauprime$.  

Next, we prove the second inequality from the claim. For a contradiction, assume that $\prof\left(\Size {X_a} \gamma\right) \geq \prof\left(\Size {\Y_a} \tau\right)$. Since $\gamma \leq \tauprime$, we further get $\prof\left(\Size{X_a}{\tauprime}\right) \geq \prof\left(\Size {\Y_a} \tau\right) =  \prof\left(\Size {\Y_a} \tauprime - g_{\Y}\right)$. That is, $\efcount\left(\Size {X_a} \tauprime, \Size {\Y_a} \tauprime\right) \leq 1$. This envy count contradicts the definition of $t$ (and, correspondingly, $\tauprime$); see equation (\ref{eq:t-repeat}). Therefore, by way of contradiction, we obtain the second part of the claim.
\end{proof}

\begin{claim}\label{goods_density_gap}
For each good $g \in \Size {X_a} \gamma$ and any good $g' \in \Y \setminus \Size {\Y_a} \tauprime$, we have $\dens_a(g) \geq \dens_a(g')$.
\end{claim}
\begin{proof}
We will use Proposition \ref{greedyproperty_gap} to prove the claim. Recall that the goods in subsets $X$ and $Y$ are indexed in the allocation order $\sigma$.  Let $t_1$ denote the index of the given good $g$ in $X$, and let $t_2$ denote the index of good $g'$ in $Y$. 

We first show that $\prof\left(X^{(t_1-1)} \right) < \prof\left(Y^{(t_2-1)} \right)$. Note that good $g \in \Size {X_a} \gamma$. Hence, all the goods included in $X$, before $g$, in the algorithm (i.e., all the goods in $X^{(t_1-1)}$) are also contained in $\Size {X_a} \gamma$. That is, $X^{(t_1-1)} \subseteq \Size {X_a} \gamma$. Using this containment, we obtain 
\begin{align}
\prof\left(X^{(t_1-1)} \right) & \leq \prof\left(\Size {X_a} \gamma\right) \nonumber \\ 
& < \prof\left(\Size {\Y_a} \tau\right) \tag{via Claim \ref{lemma5_gap}} \\
& < \prof\left( \Size {Z_a} {\tauprime}\right) \label{ineq:something}
\end{align}
Furthermore, since $g' \in \Y \setminus \Size {\Y_a} \tauprime$, we have $\prof\left(\Size {\Y_a} {\tauprime} \right) \leq \prof\left(Y^{(t_2-1)} \right)$. This bound and inequality (\ref{ineq:something}) imply $\prof\left(X^{(t_1-1)} \right) < \prof\left(Y^{(t_2-1)} \right)$. 

The above-mentioned containment also gives us $s_a \left(X^{(t_1-1)} \right) \leq s_a \left(\Size {X_a} \gamma\right) = \gamma$. Since $\gamma \leq \tauprime$ (Claim \ref{lemma5_gap}), we get that the good $g'$ can be included in the subset $X^{(t_1-1)}$ without violating agent $a$'s budget constraint: $ s_a \left(X^{(t_1-1)} + g' \right) \leq \gamma + s_a(g') \leq \tauprime + s_a(g') \leq s_a(Z) \leq B_a$. 

Now, invoking Proposition \ref{greedyproperty_gap} (with $i = t_1 - 1$ and $j = t_2 -1$), we obtain the desired inequality, $\dens_a(g) \geq \dens_a(g')$.
The claim stands proved.
\end{proof}

We are now prove Lemma \ref{lemma_gap}. 

\begin{lemma}\label{lemma_gap}
$\prof\left(X \setminus \Size {X_a} \tau\right) \geq \prof\left(\Y \setminus \Size {\Y_a} \tauprime\right)$.
\end{lemma}
\begin{proof}
Since $\Y \subseteq Y$, the good $g_{\Y}$ appears in the subset $Y$. Recall that the goods in the subsets $Z$ and $Y = \{y_1, y_2, \ldots, y_\ell \}$ are indexed according to the allocation order $\sigma$. Write $t' \in [|Y|]$ to denote the index of $g_{\Y}$ in $Y$ (i.e., $g_{\Y} = y_{t'}$). \cref{lemma5_gap} gives us 
$\prof\left(\Size {X_a} \gamma\right) < \prof\left(\Size {\Y_a} \tau\right) = \sum_{i=1}^{t-1} \prof(z_i)\le\sum_{i=1}^{t'-1} \prof(y_i)$. That is, $\prof\left(\Size {X_a} \gamma\right) < \prof\left(\Goods Y {t'-1} \right)$.  Also, by definition of $\gamma$ (equation \ref{eq:def-g_gap}), we have that the goods in $X \setminus {\Size{X_a}{\gamma}}$ (if any) have density less than $\rho(g_{\Y})$. These observations and \cref{greedyproperty_gap} imply that including $g_{\Y}$ in ${\Size{X_a}{\gamma}}$ must violate agent $a$'s budget $B_a$, i.e., it must be the case that  
\begin{align}
    \gamma + s_a(g_{\Y}) > B_a\label{eq:gz-doesnt-fit_gap}
\end{align}

Using inequality (\ref{eq:gz-doesnt-fit_gap}), we will prove that $\prof\left(\Size {X_a} \gamma \setminus \Size {X_a} \tau\right) \geq \prof\left(\Y \setminus \Size {\Y_a} \tauprime\right)$. This bound directly implies the lemma, since $\Size{X_a}{\gamma} \subseteq X$. In particular, the size of the concerned set satisfies 
\begin{align}
s_a\left(\Size {X_a} \gamma \setminus \Size {X_a} \tau\right) & = \gamma - \tau  \nonumber \\
&  = \gamma - \tauprime + s_a(g_{\Y}) \tag{$\tauprime - s_a(g_{\Y}) = \tau$} \\
& > B_a - \tauprime \tag{via inequality (\ref{eq:gz-doesnt-fit_gap})}\\
& \geq s_a\left(\Y \setminus \Size {\Y_a} \tauprime\right) \label{ineq:sizelb_gap}
\end{align}
The last inequality follows from the facts that $s_a(\Y) \leq B_a$ and $s_a( \Size {\Y_a} \tauprime) = \tauprime$. 

Furthermore, from Claim \ref{goods_density_gap} we have that each good $g \in \Size {X_a} \gamma \setminus \Size {X_a} \tau$ is denser than each good $g' \in \Y \setminus \Size {\Y_a} \tauprime$. These bounds on densities and sizes of the subsets ${\Size{X_a}{\gamma}} \setminus \Size {X_a} \tau$ and $\Y \setminus \Size {\Y_a} \tauprime$ give us $\prof\left( {\Size{X_a}{\gamma}} \setminus \Size {X_a} \tau\right) \geq \prof\left(\Y \setminus \Size {\Y_a} \tauprime\right)$. As mentioned previously, this inequality and the containment $\Size{X_a}{\gamma} \subseteq X$ imply $\prof\left(X \setminus \Size {X_a} \tau\right) \geq \prof\left(\Y \setminus \Size {\Y_a} \tauprime\right)$. The lemma stands proved.
\end{proof}

Overall, Lemma \ref{small_ef2_gap} gives us $\efcount\left(\Size {X_a} {\tau}, \Size {\Y_a} \tauprime\right) = 2$. In addition, via Lemma \ref{lemma_gap}, we have $\prof\left(X \setminus \Size{X_a}{\tau} \right) \geq \prof\left(\Y \setminus \Size {\Y_a} \tauprime\right)$. Therefore, applying \cref{envy_diff_gap}, we conclude that $\efcount(X, \Y) \leq 2$. This establishes that desired $\EFtwo$ guarantee among the agents for the allocation computed by Algorithm \ref{algo:dense_gap}.
We have already proved in \cref{envy-against-charity} that this allocation
also satisfies the $\EFone$ property for every agent against the charity.  This completes the proof of Theorem \ref{thm:densestgreedy-ef2_gap}.

\end{document}